\documentclass[11pt]{article}

\usepackage{fullpage}
\usepackage{amssymb}
\usepackage{amsmath}
\usepackage{amsfonts}
\usepackage{amsthm}
\usepackage{url}
\usepackage{graphicx}
\usepackage{subfig}
\usepackage{fancyhdr}
\usepackage{url}

\usepackage[utf8]{inputenc}
\usepackage[english]{babel}
\usepackage{hyperref}

\newcommand{\nc}{\newcommand}
\nc{\heading}[1]{\begin{center} \large \bf #1 \end{center}}

\newcommand{\oA}{\overline{A}}
\newcommand{\oB}{\overline{B}}
\newcommand{\oI}{\overline{I}}

\newcommand{\oR}{\overline{R}}

\newcommand{\Ex}{\mathsf{E}}
\renewcommand{\Pr}{\mathsf{P}}

\theoremstyle{plain}
\newtheorem{prop}{Proposition}

\newtheorem{definition}{Definition}

\usepackage{glossaries}
\makeglossaries
\newglossaryentry{latex}
{
    name=latex,
    description={Is a mark up language specially suited
    for scientific documents}
}
\newglossaryentry{maths}
{
    name=mathematics,
    description={Mathematics is what mathematicians do}
}
\begin{document}

%\copyrightnotice
%\lipsum[1-10]

\date{January 22, 2021}
\title{Stochastic Modeling of an Infectious Disease \\
 {Part III-A: Analysis of Time-Nonhomogeneous Models
 }\footnote{An earlier version of this work was presented at ITC-32 \url{https://itc32.org/keynote.html} held in Osaka, Japan on September 22-24, 2020.  For the slides \cite{kobayashi:2020s} and You Tube video, click on \url{hp.hisashikobayashi.com}.} }
\author{Hisashi Kobayashi\footnote{The Sherman Fairchild University Professor of Electrical Engineering and Computer Science, Emeritus.
Email: Hisashi@Princeton.EDU, Website: \url{hp.HisashiKobayashi.com}, Wikipedia: \url{https://en.wikipedia.org/wiki/Hisashi_Kobayashi} 
} \\
   Dept. of  Electrical Engineering \\
   Princeton University \\
   Princeton, NJ 08544, U.S.A.}

\maketitle

\begin{abstract}

We extend our BDI (birth-death-immigration) process based stochastic model of an infectious disease to time-nonhomogeneous    
\footnote{A Markov process is called time-homogeneous (or simply homogeneous), if the state transition probabilities do not depend on time. It is called nonhomogeneous, otherwise.  Note that the BDI process $I(t)$ is a Markov process.} cases, where the model parameters can change in time according to arbitrary functions, denoted $\lambda(t), \mu(t)$ and $\nu(t)$.  

In Section 1, we discuss the deterministic model, from which we derive $\oI(t)=\Ex[I(t)]$, the expected value of the infection process $I(t)$ and related stochastic processes. The function $s(t)=\int_0^t (\lambda(u)-\mu(u))\,du$ plays a central role, with or without the external arrivals.   

In Section 2, we illustrate a numerical example, by assuming a hypothetical scenario, in which a government issues a decree to its citizens to curtail their activities that may incur further infections. We show how the public's tardy response may further increase the number of infections and prolong the epidemic much longer than one might think without a quantitative analysis.

In Section 3, we seek to obtain the probability generating function for the nonhomogeneous BDI process, by solving a partial differential equation, similar to what we conducted in Part I for the homogeneous BDI process model. We find, however,  that an exact solution is obtainable only for the BD process, i.e., no arrivals of the infected from the outside. We find that the coefficient of variation for the nonhomegenous BD process to be well over unity, practically for all $t$.  This result implies that the variations among different sample paths will be as large as in the negative binomial distribution with $\nu/\lambda<1$, which we found in Part I for the homogeneous BDI model.

In the final section, we illustrate, using our running numerical example, how much information we can derive from the time-dependent PMF (probability mass function) $P_k(t)=\Pr[I(t)=k]$.  We present graphical plots of of the PMF at various $t$'s, and cross-sections of this function at various $k$'s, along the axis parallel to the $t$ axis. A mesh plot of the three dimensional array $Z(k,t)\triangleq P_k(t)$ over the $(k, t)$ plane is shown to summarize the above numerous plots.

Our analysis in the present paper reinforces our earlier claim (see \cite{kobayashi:2020b} Abstract) that it would be a futile effort to attempt to identify all possible reasons why environments of similar situations differ so much in their epidemic patterns and the number of casualties. Mere ``luck" or ``chances" play a more significant role than most of us believe. We should be prepared for a worst possible scenarios, which only a stochastic model can provide with probabilistic qualification.

An empirical validation of the above results and implications will be presented in a companion paper \cite{kobayashi:2021b}.

\end{abstract}


\tableofcontents
\section{Time-Nonhomogeneous Deterministic Model }\label{sec:nonhomo-deterministic-model}

In Parts I \& II, we assumed that any of the model parameters $\lambda, \mu$ and $\nu$ do not change in time $t$. Such models are said to be time-homogeneous, or simply homogeneous. In this Part III-A, we extend our earlier results to non-homogeneous cases.   The results we obtain for the general non-homogeneous model allow us to analyze the effectiveness of a given pandemic policy of a government and its public's response. An increase in the so-called social distances and an introduction of effective vaccines will both help reduce the value of $\lambda(t)$.  Availability of medical facilities and staff will keep the value of the function $\mu(t)$ intact, while their insufficiency  will lead to a decline or drop in the $\mu(t)$ value.  The so-called ``lock-down" policy is equivalent to an attempt to let both $\lambda(t)$ and $\nu(t)$ reduce towards zero promptly.
   
We start with generalizing the model discussed in Parts I \& II and obtain the expected value of the stochastic process $I(t)$.

\subsection{Derivation of $\oI(t)$: The expected value of the Infection Process $I(t)$}\label{subsec:derivation of oI(t)}
Recall Eqn. (22) of Part I, Section 3.2, i.e.,
\begin{align}
\frac{d\oI(t)}{dt}= a \oI(t)+\nu, ~~\mbox{where}~~a=\lambda-\mu. \label{Part_I_Eq_22}
\end{align}

Let the model parameters be generalized to arbitrary functions of time $t$, i.e.,
\begin{align}
\frac{d\oI(t)}{dt}=a(t)\oI(t)+\nu(t), \label{diff-eq}
\end{align}
where
\begin{align}
a(t)&=\lambda(t)-\mu(t). \label{a(t)}
\end{align}
The function $a(t)$ largely determines the deterministic model, i.e., the expected value of the stochastic process $I(t)$. This function can be alternatively written as
\begin{align}
a(t)=({\cal R}(t)-1)/\tau(t),
\end{align}
or
\begin{align}
a(t)=\lambda\left(1-{\cal R}^{-1}(t)\right),
\end{align}
where $\tau(t)$ is the inverse of $\mu(t)$, representing the expected period that an infected person at time $t$ remains infectious until his/her recovery, removal or death:
\begin{align}
\tau(t) \triangleq \mu(t)^{-1}.
\end{align}
and ${\cal R}(t)$ is the \emph{effective reproduction number}:  
\begin{align}
{\cal R}(t)\triangleq \frac{\lambda(t)}{\mu(t)}=\lambda(t)\tau(t).
\end{align}
The value at $t=0$, ${\cal R}(0)$,  is referred to as the \emph{basic reproduction number} (cf. \cite{kobayashi:2020a}, Eqn.(6)). 

A standard technique for solving the above differential equation is to obtain its homogeneous differential equation \footnote{A homogeneous differential equation for $y$ involves only $y$ and terms involving derivatives of $y$, such as in $\frac{d^2y}{dx^2}+a(x)\frac{dy}{dx}+b(x)y=0$. If we have $c(x)$ instead of $0$ in the RHS, it is called a non-homogeneous differential equation.} first:
\begin{align}
\frac{d\oI(t)}{dt}=a(t)\oI(t), \label{homogeneous-diff-eq}
\end{align}
which readily leads to
\begin{align}
\oI(t)=\oI(0)e^{s(t)},  \label{oI(t)}
\end{align}
where
\begin{equation}
\fbox{
\begin{minipage}{7.5cm}
\[
s(t)=\int_0^t a(u)\,du=\int_0^t(\lambda(u)-\mu(u))\,du.
\]
\end{minipage}
} \label{s(t)}
\end{equation}

Then we multiply (\ref{diff-eq}) by $e^{-s(t)}$, obtaining
\begin{align}
\frac{d\oI(t)}{dt}e^{-s(t)}=a(t)e^{-s(t)}\oI(t)+\nu(t)e^{-s(t)},
\end{align}
which can be written as
\begin{align}
\frac{d(\oI(t)e^{-s(t)})}{dt}=\nu(t)e^{-s(t)}.
\end{align}
Thus, we find
\begin{align}
\oI(t)e^{-s(t)}=\int_0^t \nu(u)e^{-s(u)}\,du + C,  \label{I(t)-and-C}
\end{align}
where the integration constant $C$ can be determined by setting $t=0$ in the above, yielding $C=\oI(0)\triangleq I_0$:
\begin{equation}
\fbox{
\begin{minipage}{5cm}
\[
   \oI(t)=I_0e^{s(t)}+e^{s(t)}N(t),
   \]
\end{minipage}
} \label{I(t)-non-homogeneous}
\end{equation}
where
\begin{align}
N(t)\triangleq \int_0^t \nu(u)e^{-s(u)}\,du, ~~t\geq 0.\label{function-N(t)}
\end{align}

If the arrival rate function is homogeneous, i.e.,
\begin{align}
\nu(t)=\nu,~~\mbox{for all}~~t\geq 0,
\end{align}
Then 
\begin{align}
N(t)=\nu\Sigma(t),~~t\geq 0,
\end{align}
where
\begin{equation}
\fbox{
\begin{minipage}{5cm}
\[
   \Sigma(t)\triangleq \int_0^t e^{-s(u)}\,du, ~~t\geq 0.
   \]
\end{minipage}
} \label{function-Sigma(t)}
\end{equation}

In the \emph{homogeneous BDI process} model studied in Parts I \& II, $a(t)=a=\lambda-\mu$ and $\nu(t)=\nu$ for all $t$. Then 
\begin{align}
s(t)=at, ~~\mbox{and}~~\Sigma(t)=\frac{1-e^{-at}}{a}.
\end{align}
Thus, (\ref{I(t)-non-homogeneous}) becomes
\begin{align}
\oI(t)=I_0 e^{at} + \frac{\nu}{a}\left(e^{at}-1\right),\label{I(t)-homogeneous}
\end{align}
which we obtained in Part I.  

The above expression can be somewhat simplified, depending on the initial value $I_0$ of infected person at time $t=0$,  and whether the security control at the boundaries is perfect or not.
\begin{enumerate}

\item \textbf{When $\mathbf{\nu(t)=0}$, and $\mathbf{I_0\geq 1}$}:\\
This corresponds to Kendall's ``The generalized birth-and-death process" \cite{kendall:1948a}, and the non-homogeneous case of the \emph{Feller-Arley} (FA) process model discussed in Part II, Section 3 \cite{kobayashi:2020b}. Equation (\ref{I(t)-non-homogeneous}) reduces to
\begin{align}
\oI(t)&=I_0 e^{s(t)}, ~~t\geq 0.\label{I(t)-FA_process}
\end{align} 

\item \textbf{When $\mathbf{\nu(t)>0}$, and $\mathbf{I_0=0}$}:\\
Then (\ref{I(t)-non-homogeneous}) reduces to
\begin{align}
\oI(t)&= e^{s(t)}N(t),   \label{I(t)-N(t)}
\end{align} 
where $N(t)$ is given by (\ref{function-N(t)}).

Let us further assume that $\nu(t)=\nu_0$ and $a(t)= a_0=\lambda_0-\mu_0$ in an initial period $[0,t_0)$ for some $t_0\approx O(a_0^{-1})$.\footnote{$O(x)$ reads as a quantity in the order of $x$.  As for a concrete value of $t_0$, see the discussion in the next section.}  Then, for $t<t_0$, $s(t)\approx a_0t$  and  $N(t)\approx\nu_0\Sigma(t)$, where
\begin{align}
\Sigma(t)= \int_0^t  e^{-s(u)}\,du\approx\frac{1-e^{-a_0t}}{a_0}, ~~0\leq t<t_0.\label{initial-N(t)}
\end{align}

If $a(t)=\lambda(t)-\mu(t)$ becomes zero at some point $t^*$ such that $\lambda(t^*)=\mu(t^*)$, and if $a(t)<0$ 
for all $t>t^*$, then eventually $s(t)$ becomes negative at some point in time, and $e^{-s(u)}$ of the integrand of (\ref{function-N(t)}) grows exponentially.  However, the multiplier $e^{s(t)}$ of ({\ref{I(t)-N(t)})  decreases negative exponentially. Consequently, $ e^{s(t)}e^{-s(u)}\ll 1$ for $t>t^*$, resulting in
\begin{align}
   \oI(t)=e^{s(t)}N(t)\approx\frac{\nu_0}{a_0} e^{s(t)} \left(1-e^{-a_0t}\right),  ~~\mbox{for all} ~~t\geq 0, ~~~\mbox{if}~~I_0=0.
\label{I(t)-approximate}
\end{align}

In the homogeneous case, the above formula reduces to
\begin{align}
\oI(t)= \frac{\nu_0}{a_0}\left(e^{a_0}-1\right),~~t>0,~~\mbox{if}~~I_0=0,
\end{align}
which, interestingly enough, makes (\ref{I(t)-approximate}) an exact formula (c.f. Eqn. (23) of Part I).

\item \textbf{When $\mathbf{\nu(t)>0}$, and $\mathbf{I_0\geq 1}$}:\\
By combining the formulas (\ref{I(t)-FA_process}) and (\ref{I(t)-approximate}), we find (\ref{I(t)-non-homogeneous}) can be approximated by
\begin{align}
 \oI(t)\approx\left(I_0+\frac{\nu_0}{a_0}\right) e^{s(t)}-\frac{\nu_0}{a_0}e^{s(t)-a_0t},  ~~\mbox{for all} ~~t\geq 0.
\label{I(t)-nonhomo-approximate}
\end{align}
For sufficiently large $t$, the second term becomes negligibly small, if $s(t)\leq a_0t$ for $t\geq 0.$}

\end{enumerate}

\subsection{Derivation of $\oA(t),\oB(t)$ and $\oR(t)$}\label{subsec:derivation-of-oA-oB-and-oR}

In Part I, Section 3.2, we defined the stochastic process $A(t)$ as the cumulative count of external arrivals of 
infected individuals from the outside. Let us assume that they arrive according to a Poisson process with rate $\nu(t)$. Then, it readily follows that
\begin{align}
\oA(t)=\Ex[A(t)]=\int_0^t\nu(u)\,du.\label{oA(t)}
\end{align}
For the time-homogeneous case, $\nu(t)=\nu$ for all $t$, the above equation reduces to
\begin{align}
\oA(t)=\nu t.
\end{align}

The stochastic process $B(t)$ is defined as the cumulative account of internally infected individuals, and each infected (and infectious, as well) person will reproduce a new infection at the rate $\lambda(t)$ at time $t$.  Then $\oB(t)=\Ex[B(t)]$ should satisfy the following differential equation
\begin{align}
\frac{d\oB(t)}{dt}=\lambda(t)\oI(t), \label{diff-eq-for-B(t)}
\end{align}
hence,
\begin{align}
\oB(t)=\int_0^t \lambda(u)\oI(u)\,du.\label{oB(t)}
\end{align}

We defined the stochastic process $R(t)$ as the cumulative count of recovered/removed/dead (hence no longer infectious) persons. Each infected person will join this group at the rate of $\mu(t)$.
Thus, we have, similarly to (\ref{diff-eq-for-B(t)}),
\begin{align}
\frac{d\oR(t)}{dt}=\mu(t)\oI(t), \label{diff-eq-for-R(t)}
\end{align}
hence,
\begin{align}
\oR(t)=\int_0^t \mu(u)\oI(u)\,du.\label{oR(t)}
\end{align}

We now consider the three different situations corresponding to those discussed in the previous section:

\begin{enumerate}
\item \textbf{When $\mathbf{\nu(t)}$,  and $\mathbf{I_0\geq 1}$}:\\
In this case we have
\begin{align}
\oB(t)-\oR(t)&=I_0\int_0^t(\lambda(u)-\mu(u))e^{s(u)}\,du=I_0\int_0^t s'(u)e^{s(u)}\,du\nonumber\\
&=I_0\int_0^t\left(e^{s(u)}\right)'\,du=I_0\left[e^{s(t)}-e^{s(0)}\right]=I_0\left(e^{s(t)}-1\right),
\end{align}
from which we find
\begin{align}
\oI(t)=I_0+\oB(t)-\oR(t).
\end{align}

\item \textbf{When $\mathbf{\nu(t)>0}$,  and $\mathbf{I_0=0}$}:\\
In this case,
\begin{align}
\oB(t)-\oR(t)&=\int_0^t (\lambda(u)-\mu(u))e^{s(u)}N(u)\,du\nonumber\\
&=\int_0^t\left(e^{s(u)}\right)'N(u)\,du=\left[e^{s(u)}N(u)\right]_0^t-\int_0^te^{s(u)}N'(u)\,du\nonumber\\
&=e^{s(t)}N(t)-\int_0^t e^{s(u)}\nu(u)e^{-s(u)}\,du=\oI(t)-\oA(t),
\end{align}
from which we obtain
\begin{align}
\oI(t)=\oA(t)+\oB(t)-\oR(t),
\end{align}

\item \textbf{When $\mathbf{\nu(t)>0}$,  and $\mathbf{I_0\geq 1}$}:\\
For this general case, we have
\begin{align}
\oB(t)-\oR(t)&=\int_0^t s(u)'\left[I_0e^{s(u)}+e^{s(u)}N(u)\right]\,du\nonumber\\
&=I_0\left(e^{s(t)}-1\right)+e^{s(t)}N(t)-\oA(t),
\end{align}
from which we find
\begin{align}
I_0e^{s(t)}+e^{s(t)}N(t)=I_0+\oA(t)+\oB(t)-\oR(t).
\end{align}
Since the LHS is $\oI(t)$ as given in (\ref{I(t)-non-homogeneous}), we have 
\begin{align}
\oI(t)=I_0+\oA(t)+\oB(t)-\oR(t),
\end{align}
which could have been directly obtained  from the identity (18) of Part I, Section 3.2.

\end{enumerate}

\subsection{Daily Counts of the Infected, Recovered and Dead}
Recall that the process of our interest  $I(t)$  is the number of currently infected individuals, excluding those who have recovered, removed (to e.g., hospitals) or have died.  Since it represents the current total \emph{infectious} individuals, it contains the most important information concerning the current and future infections.  In practice, however, the statistics that are most frequently reported in mass media are:
\begin{itemize}
\item[(i)] Daily counts of newly infected persons;
\item[(ii)] Cumulative count of infected persons up to the present;
\item[(iii)] Daily counts of newly died persons;
\item[(iv)] Cumulative count of deaths up to the present;
\item[(v)] Daily counts of persons who are seriously ill and treated in hospitals, etc.  
\end{itemize}

In our companion paper \cite{kobayashi:2021b}, which report on our simulation study, we will show some of these statistics in terms of bar charts. Thus, it will be instructive to derive the expected values of some of these statistics of interest.

\begin{definition}[New Infections]

$\oI_{new}[t]$ is defined as the expected number of newly infected persons on day $t$, $t=0, 1, 2, \cdots$. 
\end{definition}
\begin{definition}[New Recoveries]

$\oR_{new}[t]$ is defined as the expected number of newly recovered persons on day $t$, $t=0, 1, 2, \cdots$.   
\end{definition}

With these definitions we state the following simple formulas as a proposition:

\begin{prop}[{Formulas for $\oR_{new}[t]$ and $\oI_{new}[t]$}]

$\oR_{new}[t]$ is given by 
\begin{align}
\oR_{new}[t]=\int_{t-1}^t \mu(u)\oI(u)\,du=\oR(t)-\oR(t-1),~~ t=1, 2, 3, \cdots.\label{oR_new[t]}
\end{align}
Similarly, $\oI_{new}[t]$ is given by
\begin{align}
\oI_{new}[t]=\int_{t-1}^t \lambda(u)\oI(u)\,du +\int_{t-1}^t\nu(u)\,du =\oB(t)-\oB(t-1)+\oA(t)-\oA(t-1),~~t=1, 2, 3,  \cdots.
\label{oI_new[t]}
\end{align}
\end{prop}
\begin{proof}
The above formulas are readily found from the  computation of $\oA(t), \oB(t)$ and $\oR(t)$ given in (\ref{oA(t)}), (\ref{oB(t)}) and (\ref{oR(t)}) of Section \ref{subsec:derivation-of-oA-oB-and-oR}
\end{proof}

\section{Application of the Time-Nonhomogeneous Deterministic Model}\label{sec:Example-nonhomogenous-deterministic}
 
In this example, we present an illustrative example how the above general analytic results can be utilized by considering a situation in which a government declares a state of emergency, and requests its citizens to substantially curtail their activities so as to significantly reduce the infectious rate $\lambda(t)$.  The so-called lock-down corresponds to making $\lambda(t)\approx 0$ by banning any contacts with individuals outside the household.

\subsection{Computation of $\oI(t)$}\label{subsec:compute-oI}
Consider the infection rate function $\lambda(t)$ depicted in Figure \ref{fig:lambda(t)-raised-cosine}, in which $\lambda(t)$ takes on a constant value $\lambda_0$ during the initial period $0\leq t\leq t_1(=50)$. We assume that at time $t_1$ the government issues a decree to its citizens to significantly reduce its social contacts.  Not all public may respond to the decree promptly, so we assume that it takes $d$ days to implement the order.   Between $t_1$ and $t_{1d}=t_1+d$, $\lambda(t)$ decreases monotonically, then takes on another constant value $\lambda_1$ for $t\geq t_{1d}$.  

\begin{align}
\lambda(t)&=\left\{\begin{array}{ll}
\lambda_0,~~&\mbox{for} ~~0\leq x \leq t_1,\\
\lambda_1+\frac{(\lambda_0-\lambda_1)}{2}\left(1+\cos\frac{\pi (t-t_1)}{d}\right),
%=\frac{a_0+a_1}{2}+\frac{(a_0-a_1)}{2} \cos\frac{\pi(t-t_1)}{d},
~~&\mbox{for}~~t_1\leq t\leq t_{1d},\\
\lambda_1,~~&\mbox{for}~~t\geq t_{1d}.\end{array} \right.
\end{align}
By adopting a ``raised-cosine" curve between $t_1$ and $t_{1d}$, we have smooth connections at both ends: $t_1$ and $t_{1d}$, but the shape of this transitional curve is not as important as the delay $d$ value. 

As we see in the rest of this section even a small delay in implementing the government's new guideline may significantly impact the effectiveness of the decree. We consider three cases: $d=0, 5$, and $10$~[days], and the consequences of the different delays are shown in three different colors; cyan, red and blue, respectively.

\begin{figure}[thb]
\begin{minipage}[t]{0.45\textwidth}
\centering
\includegraphics[width=\textwidth]{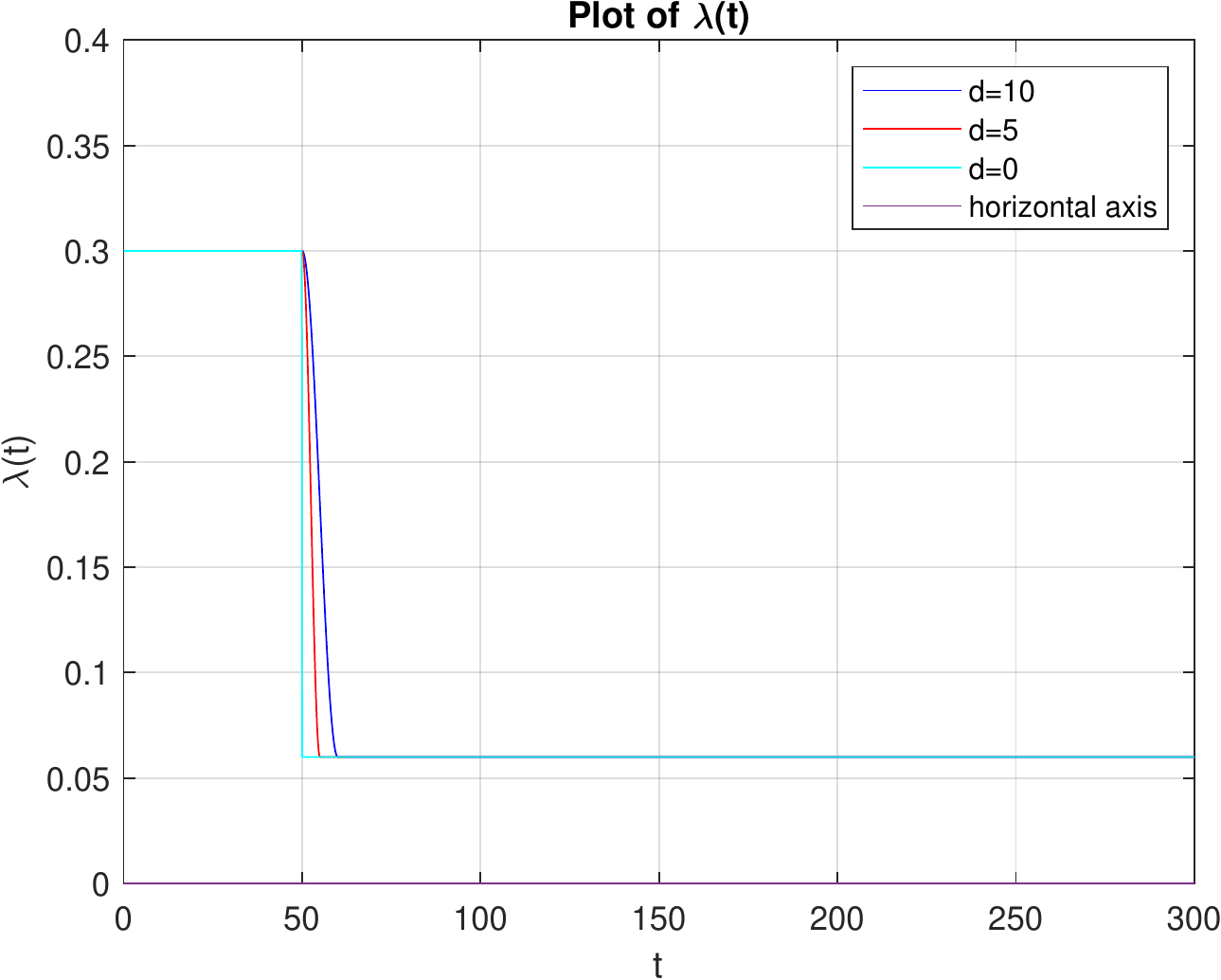}
\caption{\sf The function $\lambda(t)$ makes a transition from $\lambda_0=0.3$ down to $\lambda_1=0.06$ over the interval $t_1=50$ to $t_{1d}=50+d$ where three different delays are considered.}
\label{fig:lambda(t)-raised-cosine}
\end{minipage} 
\qquad
\begin{minipage}[t]{0.45\textwidth}
\centering
\includegraphics[width=\textwidth]{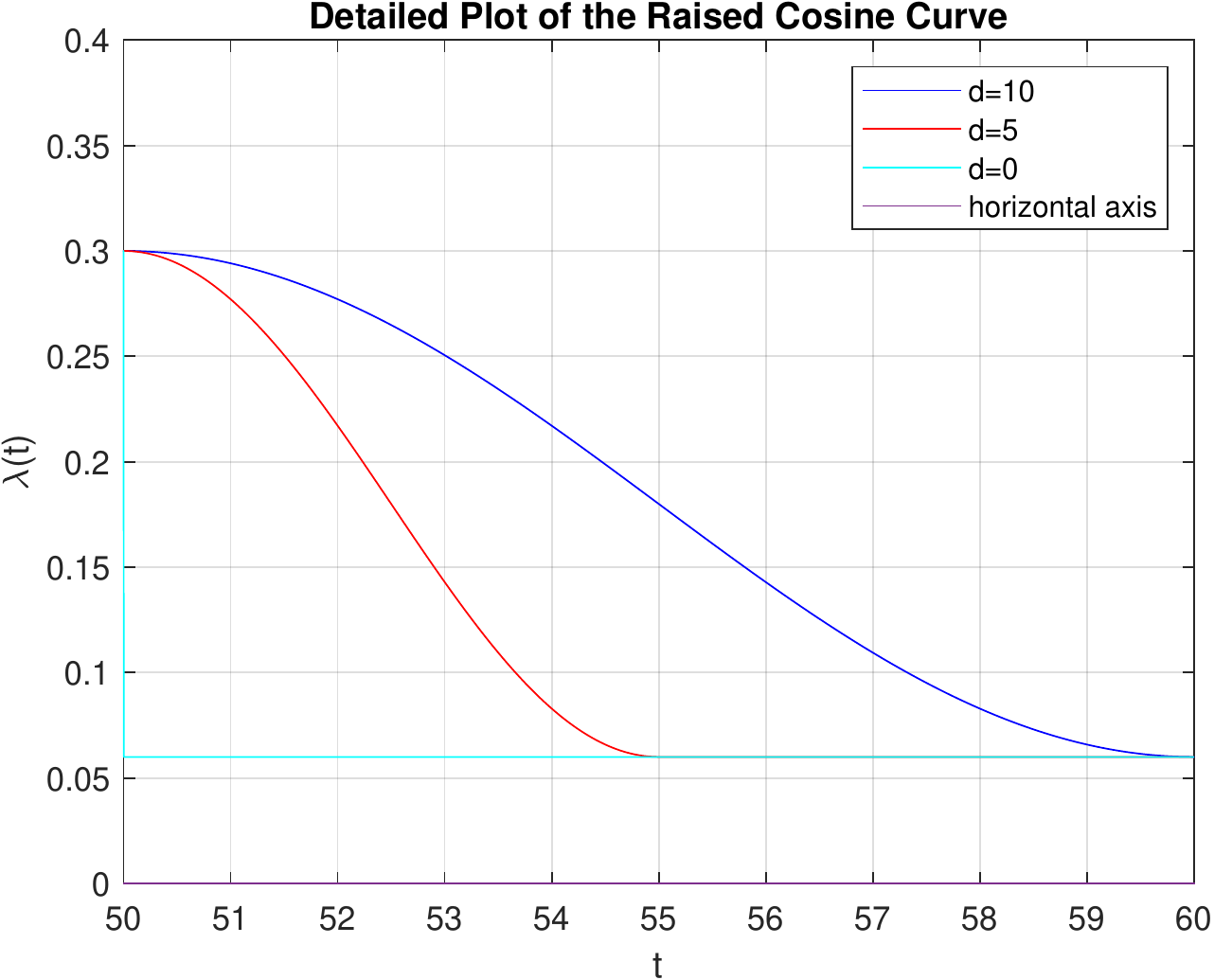}
\caption{\sf An expanded view of the transition, where the function $\lambda(t)$ takes a smooth curve represented one half cycle of a cosine function, raised up properly .} 
\label{fig:Detailed-raised-cosine-bent-at-t=50}
\end{minipage}
\end{figure}

The corresponding $s(t)$ of (\ref{s(t)}) is obtained by integrating $a(t)=\lambda(t)-\mu(t)$.  In this example,
 $\mu(t)=0.1$ for all $t$:
\begin{align}
s(t)&=\left\{\begin{array}{ll}
a_0 t,~~&\mbox{for} ~~0\leq t \leq t_1,\\
\alpha+\beta t +\gamma\sin\theta(t),~~&\mbox{for} ~~t_1\leq t \leq t_{1d},\\
s(t_{1d})+a_1(t-t_{1d}).~~&\mbox{for}~~t\geq t_{1d},\end{array} \right.\label{s(t)-example}
\end{align}
where
\begin{align}
a_0&=0.3-0.1=0.2,~~\mbox{and}~~a_1=0.06-0.1=-0.04.\label{a_0-etc}
\end{align}
and 
\begin{align}
\alpha &=\frac{(a_0-a_1)t_1}{2},~~\beta=\frac{(a_0+a_1)}{2},~~\gamma =\frac{(a_0-a_1)d}{2\pi},~~\mbox{and}~~
\theta(t)=\frac{\pi(t-t_1)}{d}. \label{alpha-etc}
\end{align}
In Figure \ref{fig:s(t)-bent-at-t=50} we plot the above $s(t)$. The fact that $s(t)$  does not become negative until around $t=300$ implies that the epidemic does not begin to die down completely until that period.

We discuss the behavior $\oI(t)$ for the three different cases defined in the previous sections.

\begin{enumerate}

\item \textbf{When $\mathbf{\nu(t)=0}$, and $\mathbf{I_0\geq 1}$}:\\

When $\nu(t)=0$, our model reduces to the BD (birth-and-death) process \footnote{The BD process is often referred to as  \emph{Feller-Arley} (FA) process.} model. The function $s(t)$ defined by (\ref{s(t)}) is plotted in Figure \ref{fig:s(t)-bent-at-t=50} and  $\oI(t)$ of (\ref{I(t)-FA_process}), with $I_0=1$,  is shown in Figure \ref{fig:I(t)-bent-at-t=50-and-I_0=1}.
\begin{figure}[thb]
\begin{minipage}[t]{0.45\textwidth}
\centering
\includegraphics[width=\textwidth]{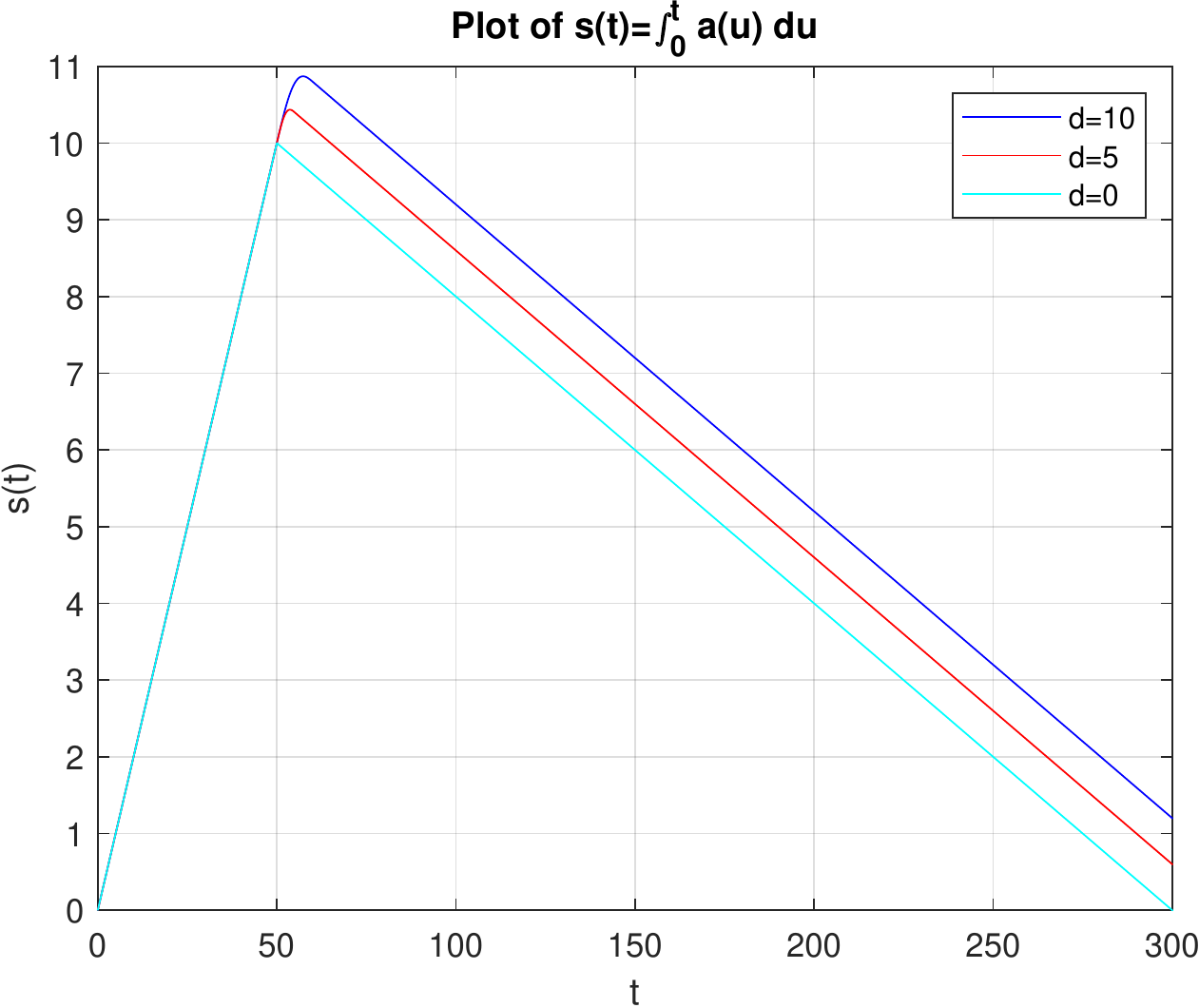}
\caption{\sf The function $s(t)=\int_0^t (\lambda(u)-\mu(t))du$: where $\lambda(t)$ is given in Figure \ref{fig:lambda(t)-raised-cosine} and $\mu(t)=0.1$ for all $t$.} 
\label{fig:s(t)-bent-at-t=50}
\end{minipage}
\qquad
\begin{minipage}[t]{0.45\textwidth}
\centering
\includegraphics[width=\textwidth]{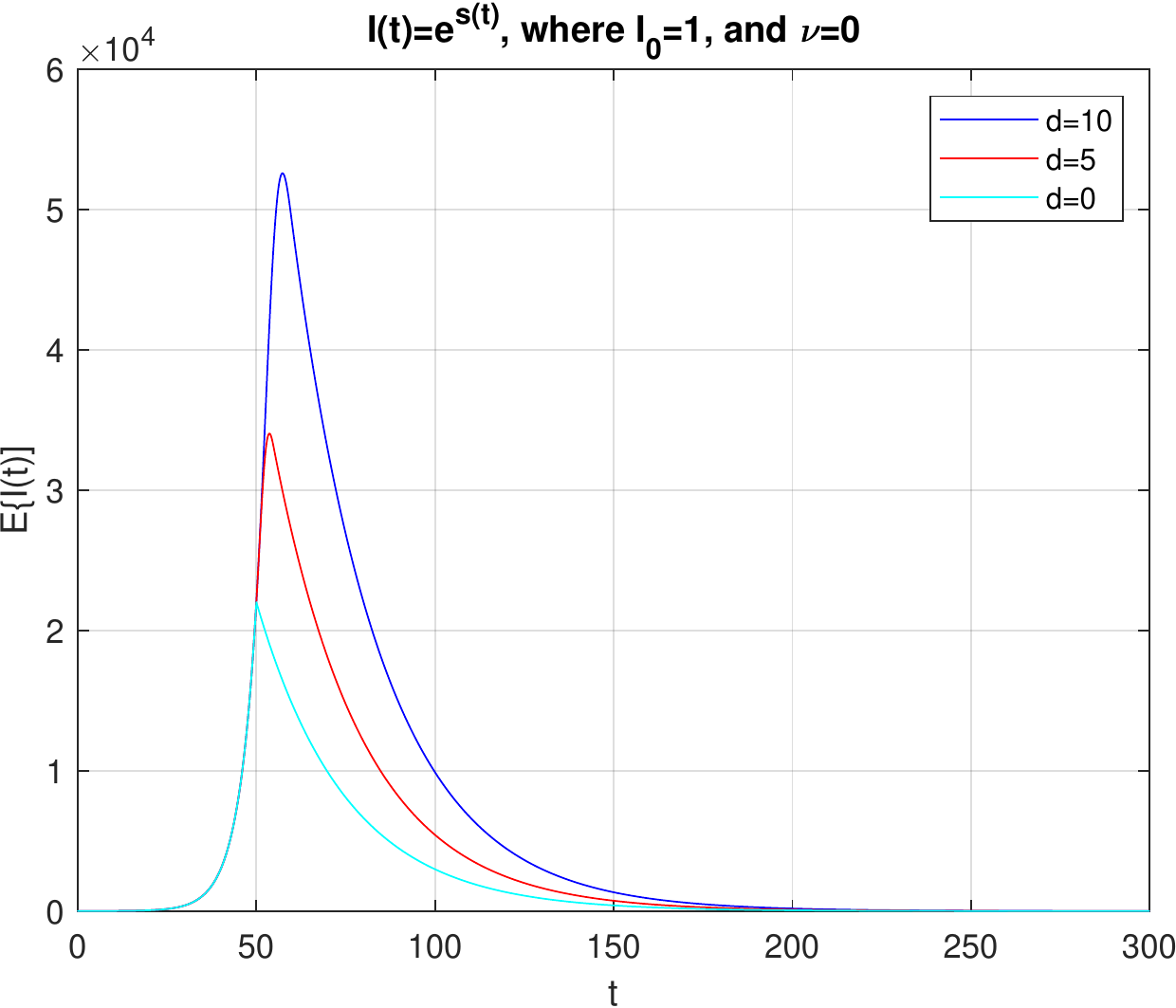}
\caption{\sf The function $\oI(t)=I_0 e^{s(t)}$, with $\nu(t)=0$ and $I_0=1$: where $s(t)$ is given in Figure \ref{fig:s(t)-bent-at-t=50}.} 
\label{fig:I(t)-bent-at-t=50-and-I_0=1}
\end{minipage} 
\end{figure}

\begin{figure}[thb]
\begin{minipage}[t]{0.45\textwidth}
\centering
\includegraphics[width=\textwidth]{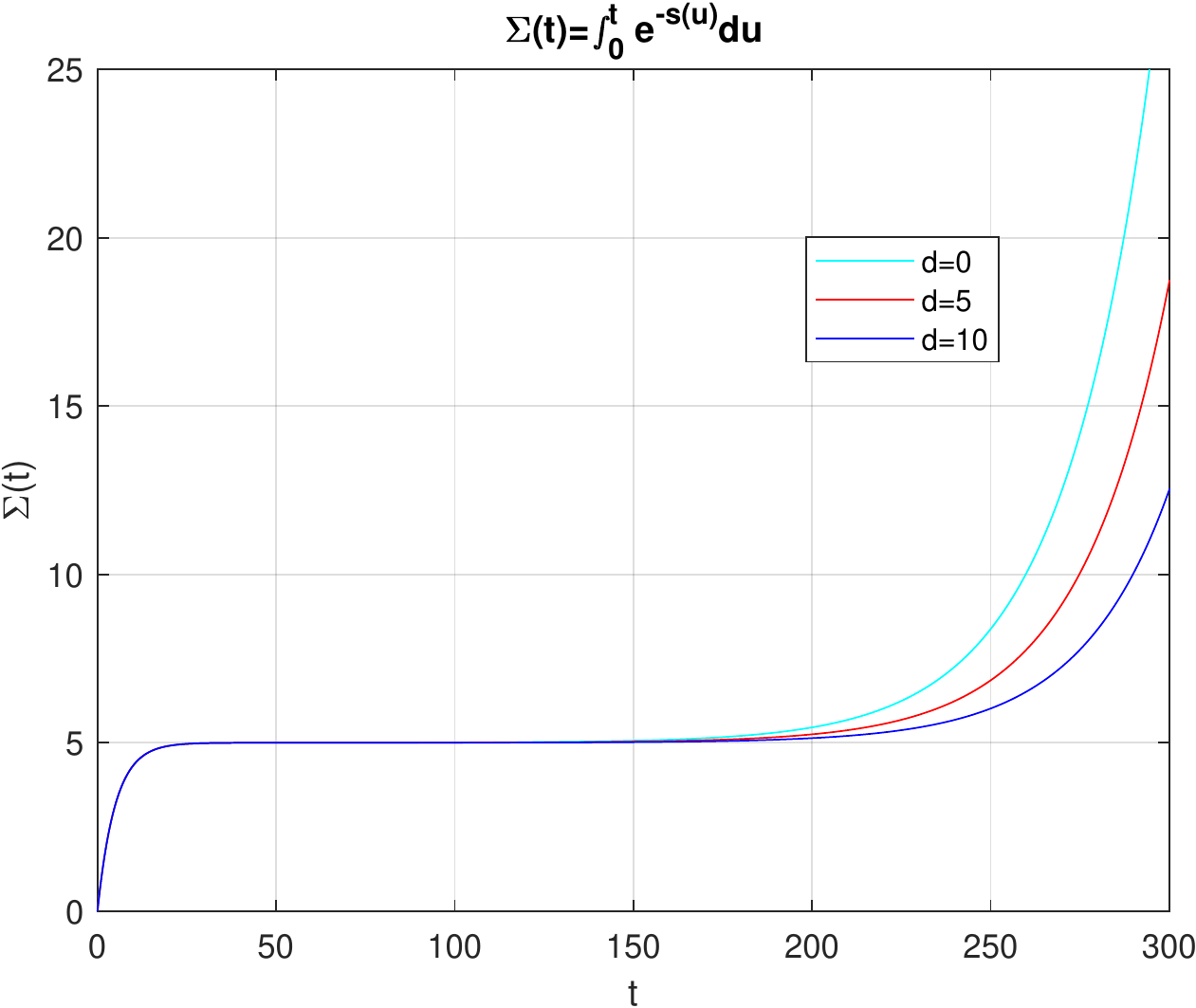}
\caption{\sf Plot of $\Sigma(t)=\int_0^t e^{-s(u)}\,du$ of (\ref{function-Sigma(t)}).} 
\label{fig:Sigma(t)}
\end{minipage}
\qquad
\begin{minipage}[t]{0.45\textwidth}
\centering
\includegraphics[width=\textwidth]{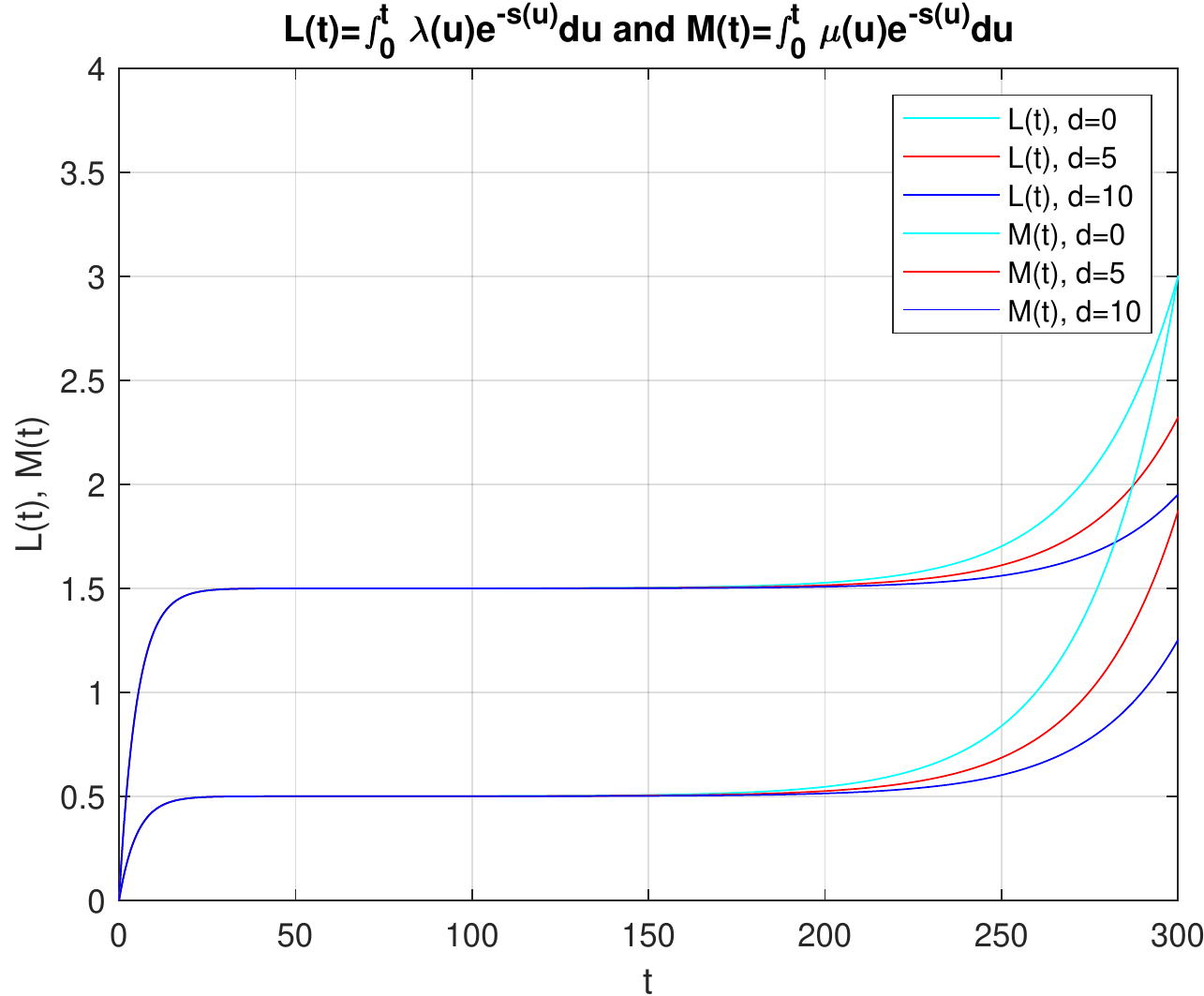}
\caption{\sf Plot of $L(t)$ and $M(t)$. }
\label{fig:L(t)-and-M(t)-bent-at-t=50}
\end{minipage} 
\end{figure}
\begin{figure}
\begin{minipage}[t]{0.45\textwidth}
\centering
\includegraphics[width=\textwidth]{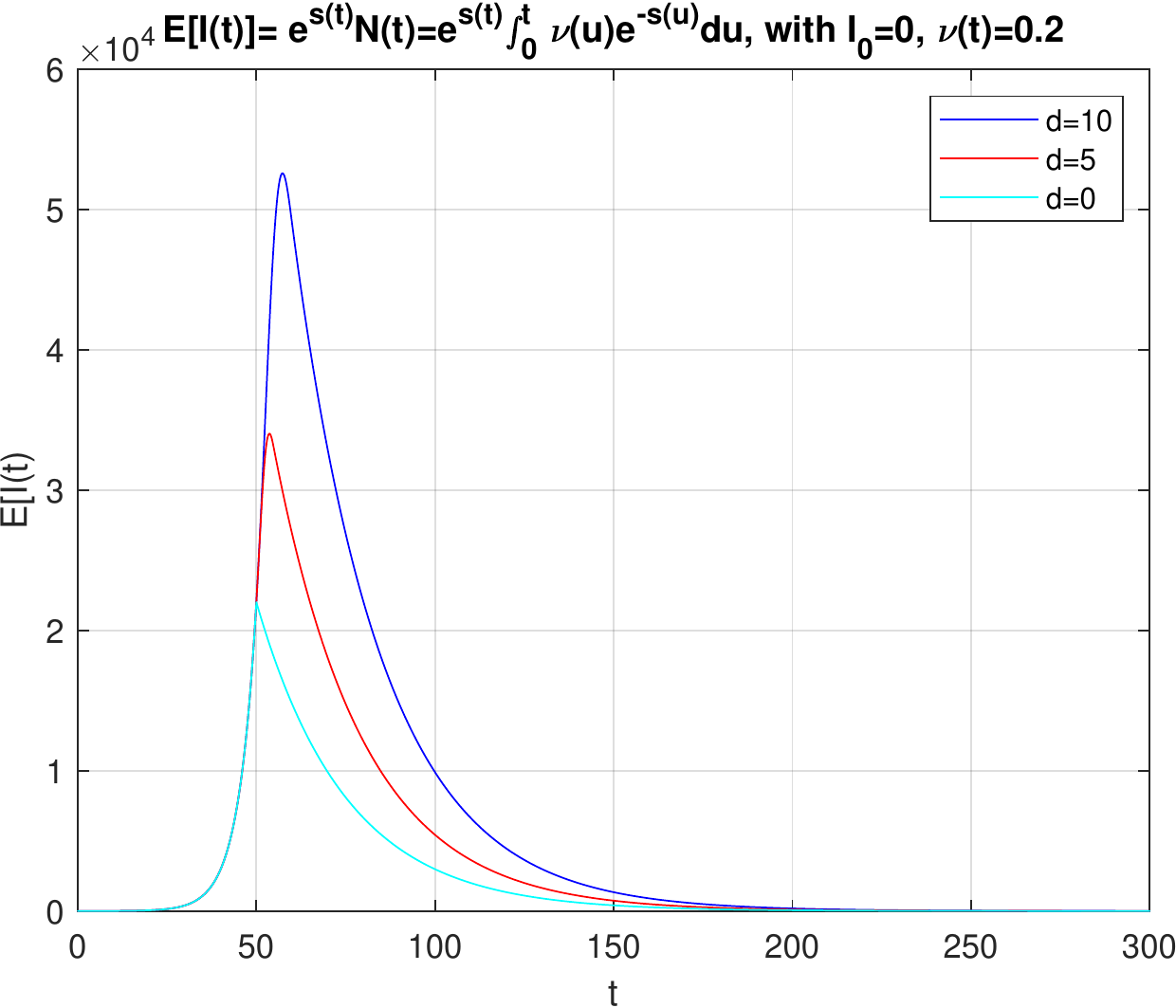}
%\caption{\sf Plot of $\oI(t)$ of (\ref{eq:I(t)-k_0=1-I_0=1}), i.e., $\nu>0$ and $I_0=0$.}
\caption{\sf Plot of $\oI(t)$ of (\ref{I(t)-approximate}), i.e.,$\nu>0$ and $I_0=0$,}
\label{fig:I(t)-using-N(t)-bent-at-t=50}
\end{minipage}
\qquad
\begin{minipage}[t]{0.45\textwidth}
\centering
\includegraphics[width=\textwidth]{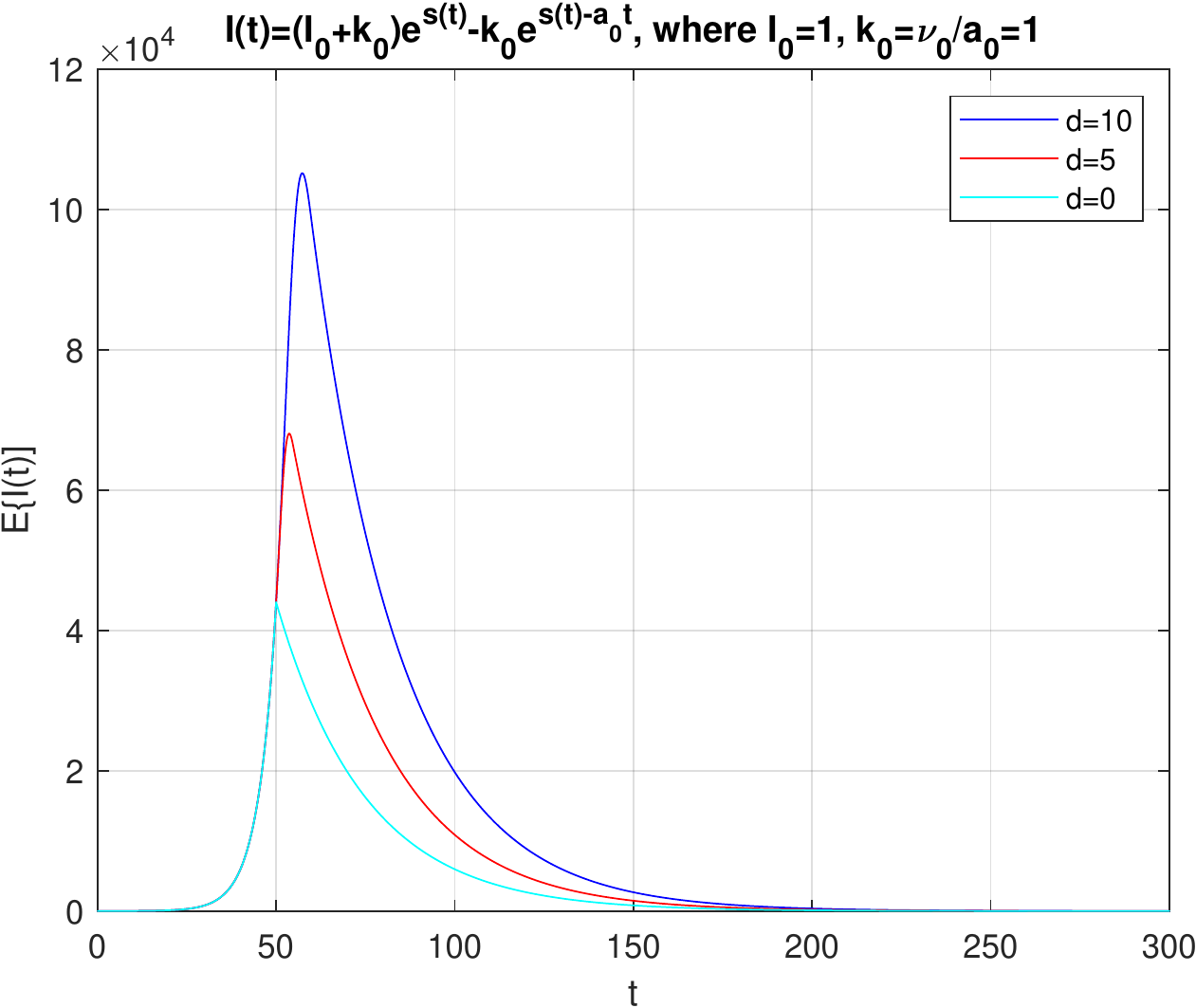}
%\caption{\sf Plot of $\oI(t)=e^{s(t)}N(t)$, with $\nu(t)=\nu_0$ and $I_0=1$.}
\caption{\sf Plot of $\oI(t)$ of (\ref{eq:I(t)-k_0=1-I_0=1}), with $\nu(t)=\nu_0$ and $I_0=1$.}
\label{fig:I(t)-k_0=1-I_0=1-bent-at-t=50}
\end{minipage}
\end{figure}

By comparing the three curves in Figure \ref{fig:I(t)-bent-at-t=50-and-I_0=1}, we note  
\begin{itemize}
\item If the public immediately respond to the government request by reducing the infection rate  $\lambda(t)$ below
$\mu(t)$ immediately (i.e., $d=0$),  the $\oI(t)$ immediately begins to decrease, as shown  by the curve in cyan. 
\item If the public takes some time, however, $\oI(t)$ continues to increase until $a(t)$ turns negative (i.e., $\lambda(t) <\mu_0=0.1$ in Figure \ref{fig:lambda(t)-raised-cosine}), when $s(t)$ starts decreasing.
\end{itemize}

\item \textbf{When $\mathbf{\nu(t)>0}$, and $\mathbf{I_0=0}$}:\\ 

Recall that $\oI(t)$ is given by (\ref{I(t)-N(t)}), and $\Sigma(t)$ defined by (\ref{function-Sigma(t)}) is plotted in Figure \ref{fig:Sigma(t)}.  In the approximation (\ref{initial-N(t)}) we somewhat vaguely defined $t_0=O(a_0^{-1})$. Let us note, for example, that $\exp(-5)\approx 0.0067$. Then the approximation (\ref{I(t)-approximate}) introduces an error of less than 0.67\% if we set $a_0t_0>5$, i.e., $t_0>5/a_0=25$.  

Noting that the function $\Sigma(t)$ reaches its plateau level $\frac{1}{a_0}$ by $t_0 (\approx 25)$, and remains flat until $t$ further increases to $t^+\approx 200$ as seen in Figure \ref{fig:Sigma(t)}.  This is because in the interval $(t_0, t^+)$,  we find that $s(t)> 5$, which makes $e^{-s(t)}<0.0067$. Thus, the contribution of the integrand $e^{-s(u)}, u\in(t_0, t^+)$ to $\Sigma(t)$ is negligibly small.  Thus, we find that the approximation (\ref{initial-N(t)}) is valid in $[0, t^+)$ for any $\lambda(t), \mu(t)$, and $\nu(t)$, so long as these functions remain constants during the initial period $[0, t_0)$.

For $u>t^+$, the function $e^{-s(u)}$ grows exponentially, as $s(u)$ continues decreasing, and eventually negative for $u>300$, as seen in Figure \ref{fig:s(t)-bent-at-t=50}. In this region, however, the function $e^{s(t)}$ rapidly  decays towards zero.  Consequently,  we obtained the approximation (\ref{I(t)-approximate})

In Figure \ref{fig:I(t)-using-N(t)-bent-at-t=50} we plot $\oI(t)$ of (\ref{I(t)-approximate}). Because of the behavior of $N(t)$ discussed earlier, the shape of this $\oI(t)$ is indistinguishable from that of Figure \ref{fig:I(t)-bent-at-t=50-and-I_0=1}; the functional form $\exp(s(t))$ essentially
determines the shape of the $\oI(t)$'s in both figures.  Their magnitudes happen to be the same, because $I_0$ and $\nu_1/a_1$ are both equal to unity in this particular example.  Needless to say, if we set $\nu_0=0.1$, for instance the $\oI(t)$ of Figure \ref{fig:I(t)-using-N(t)-bent-at-t=50} will be scaled down to one half.

\item \textbf{When $\mathbf{\nu(t)>0}$ and $\mathbf{I_0\geq 1}$}:\\

For the numerical value of the running example, with $I_0=1$ and $k_0\triangleq\frac{\nu_0}{a_0}=\frac{0.2}{0.3-0.1}=1$, we have
\begin{align}
\oI(t)\approx (I_0+k_0)e^{s(t)} -e^{s(t)-a_0t}=2e^{s(t)} -e^{s(t)-0.3t},\label{eq:I(t)-k_0=1-I_0=1}
\end{align}
where
\begin{align}
s(t)=\int_0^t (\lambda(u)-0.1)\,du.
\end{align}
The function $s(t)$ is given by (\ref{s(t)-example}) with numerical values of (\ref{a_0-etc}) and (\ref{alpha-etc}).
Figure \ref{fig:I(t)-k_0=1-I_0=1-bent-at-t=50} is a plot of $\oI(t)$ of (\ref{eq:I(t)-k_0=1-I_0=1}), and is clear that it is a sum of $\oI(t)$ plotted in Figures \ref{fig:I(t)-bent-at-t=50-and-I_0=1} and \ref{fig:I(t)-using-N(t)-bent-at-t=50}.  In this example $s(t)\leq a_0t$ for all $t\geq 0$, thus $e^{s(t)-a_ot}\leq 1$ for $t\geq 0$.\footnote{This inequality must hold under all practical situation,  where an effort to decrease $a(t)$ from the original value $a_0$ is in effect.} Thus, $\oI(t)$ of (\ref{eq:I(t)-k_0=1-I_0=1}) is, for all practical purposes, just  $(1+\frac{k_o}{I_0})=2$ times of $\oI(t)$ given by (\ref{I(t)-FA_process}), as shown in Figures \ref{fig:I(t)-bent-at-t=50-and-I_0=1} and \ref{fig:I(t)-k_0=1-I_0=1-bent-at-t=50}.

\end{enumerate}

\subsection{Computation of $\oA(t), \oB(t)$ and $\oR(t)$}\label{subsec:example-compute-oA-oB-oR}

For $\nu(t)=\nu_0$, $\oA(t)$ is simply given by
\begin{align}
\oA(t)=\nu_0 t=0.2t, ~~\mbox{for all}~~t.
\end{align}
Figures \ref{fig:B(t)} \& \ref{fig:R(t)} show $\oB(t)$ and $\oR(t)$ given by (\ref{oB(t)}) and (\ref{oR(t)}), respectively.   

Given  $\oA(t),\oB(t), \oR(t)$ obtained above and the initial condition $I_0=0$, we compute $\oI_c(t)$  
\begin{align}
\oI_c(t)=\oA(t)+\oB(t)-\oR(t) \label{I_c(t)},
\end{align}
to check the consistency among the the three stochastic processes.  The $\oI_c(t)$ computed above  should agree to $\oI(t)$ originally computed by (\ref{I(t)-approximate}), which indeed can be numerically verified.

\begin{figure}[thb]
\begin{minipage}[t]{0.45\textwidth}
\centering
\includegraphics[width=\textwidth]{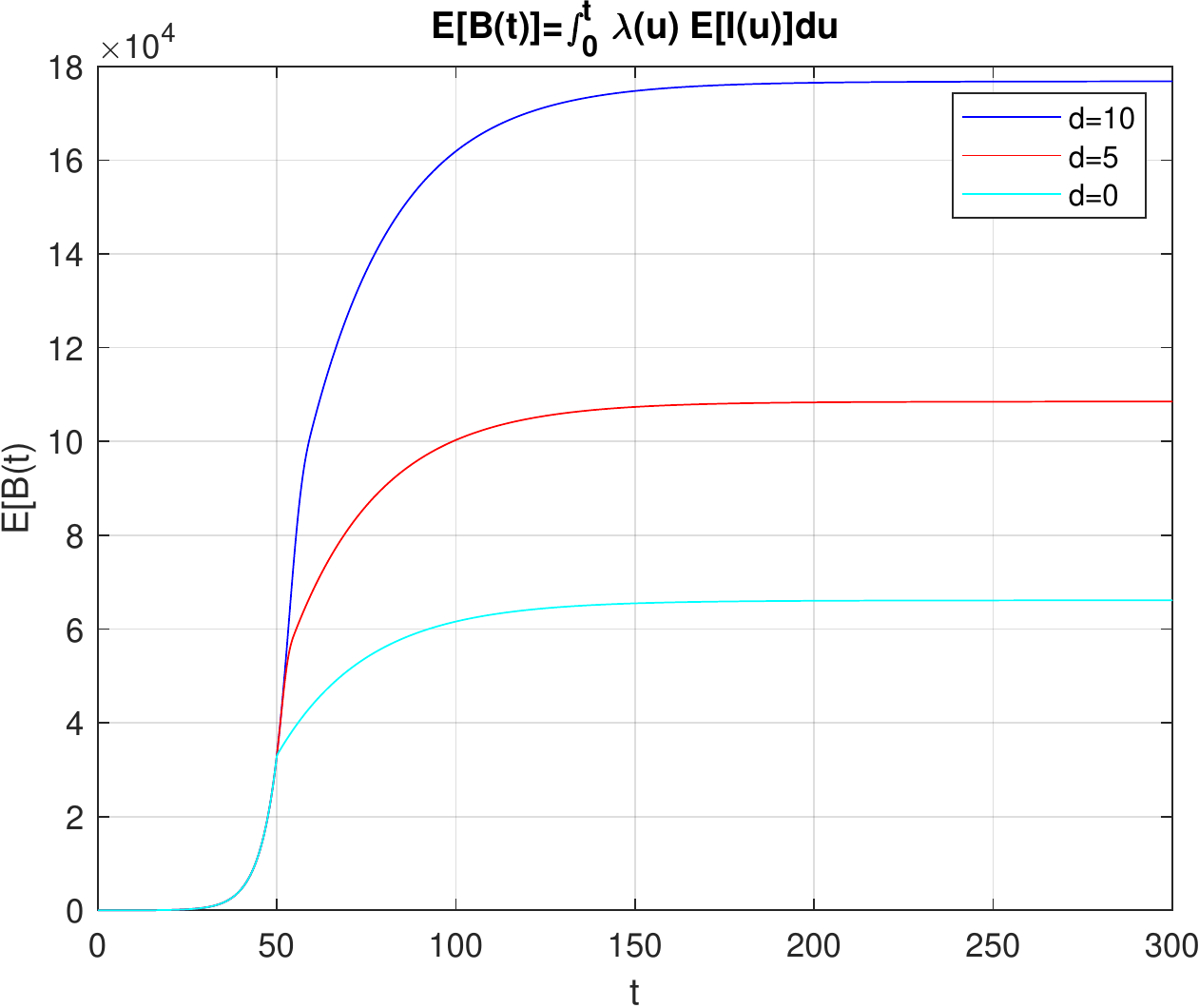}
\caption{\sf $\oB(t)$, the expected cumulative number of infections.}
\label{fig:B(t)}
\end{minipage}
\hspace{0.5cm}
\begin{minipage}[t]{0.45\textwidth}
\centering
\includegraphics[width=\textwidth]{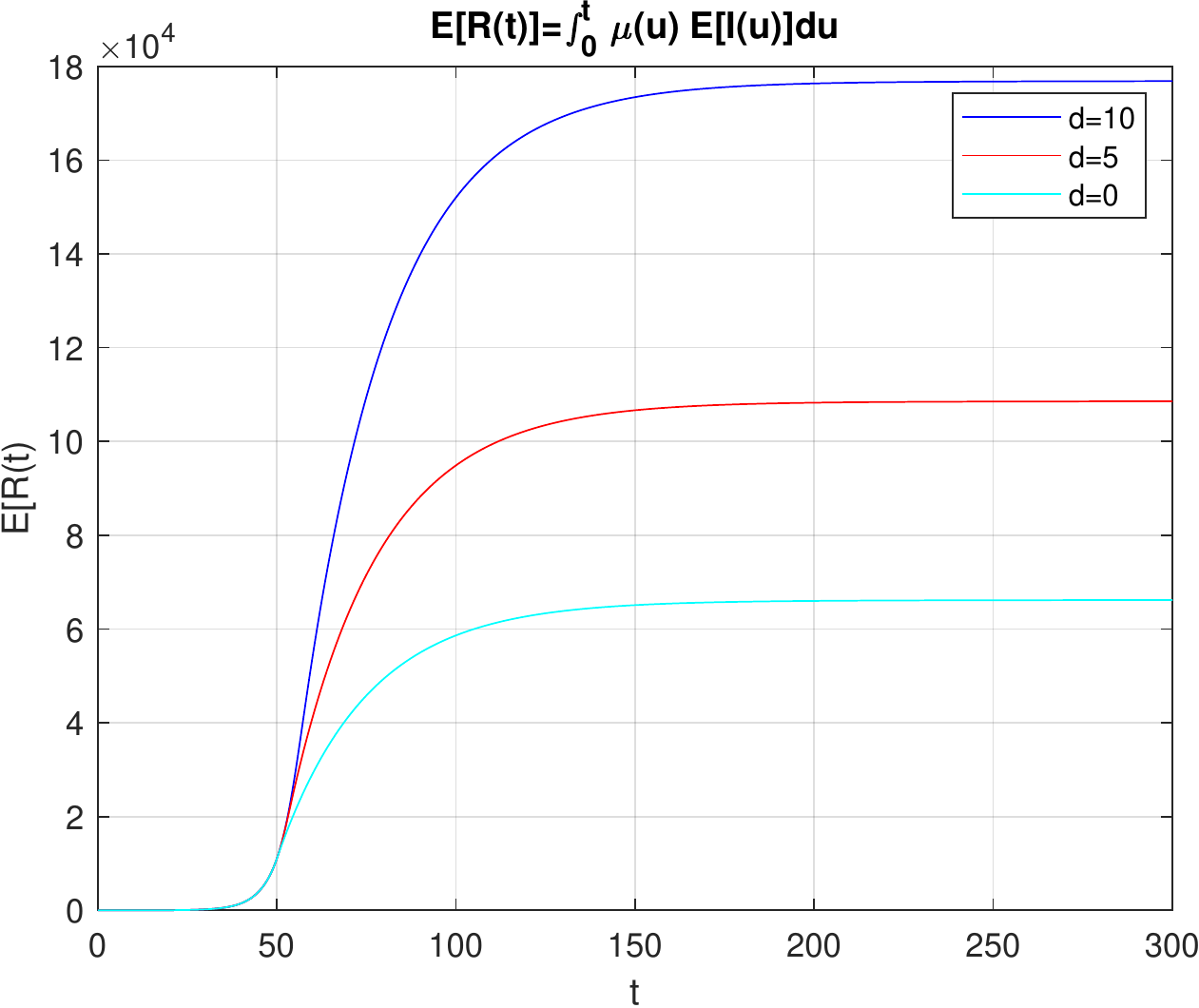} 
\caption{\sf $\oR(t)$, the expected cumulative number of recoveries. }
\label{fig:R(t)}
\end{minipage} 
\end{figure}

\clearpage
\subsection{Computation of Daily Counts of the Infected, Recovered and Dead}
\label{subsec: Comp-new-infections}

From the formulas (\ref{oR_new[t]}) and (\ref{oI_new[t]}), we can readily obtain the new infections and new recoveries, as shown in Figures \ref{fig:Example-new-infections-per-day} and \ref{fig:Example-new-recoveries-per-day}.

Note that the shape of $\oR_{new}[t]$ is proportional to $\oI(t)$, since $\mu(t)$ is constant in this running example, whereas $\lambda(t)$ changes value from $\lambda_0$ to $\lambda_1$ during the interval $t\in [50, 50+d1)$, where $d=0$ (cyan), $d=5$ (red) and $d=10$ (blue). 

\begin{figure}[hbt]
\begin{minipage}[h]{0.45\textwidth}
\centering
\includegraphics[width=\textwidth]{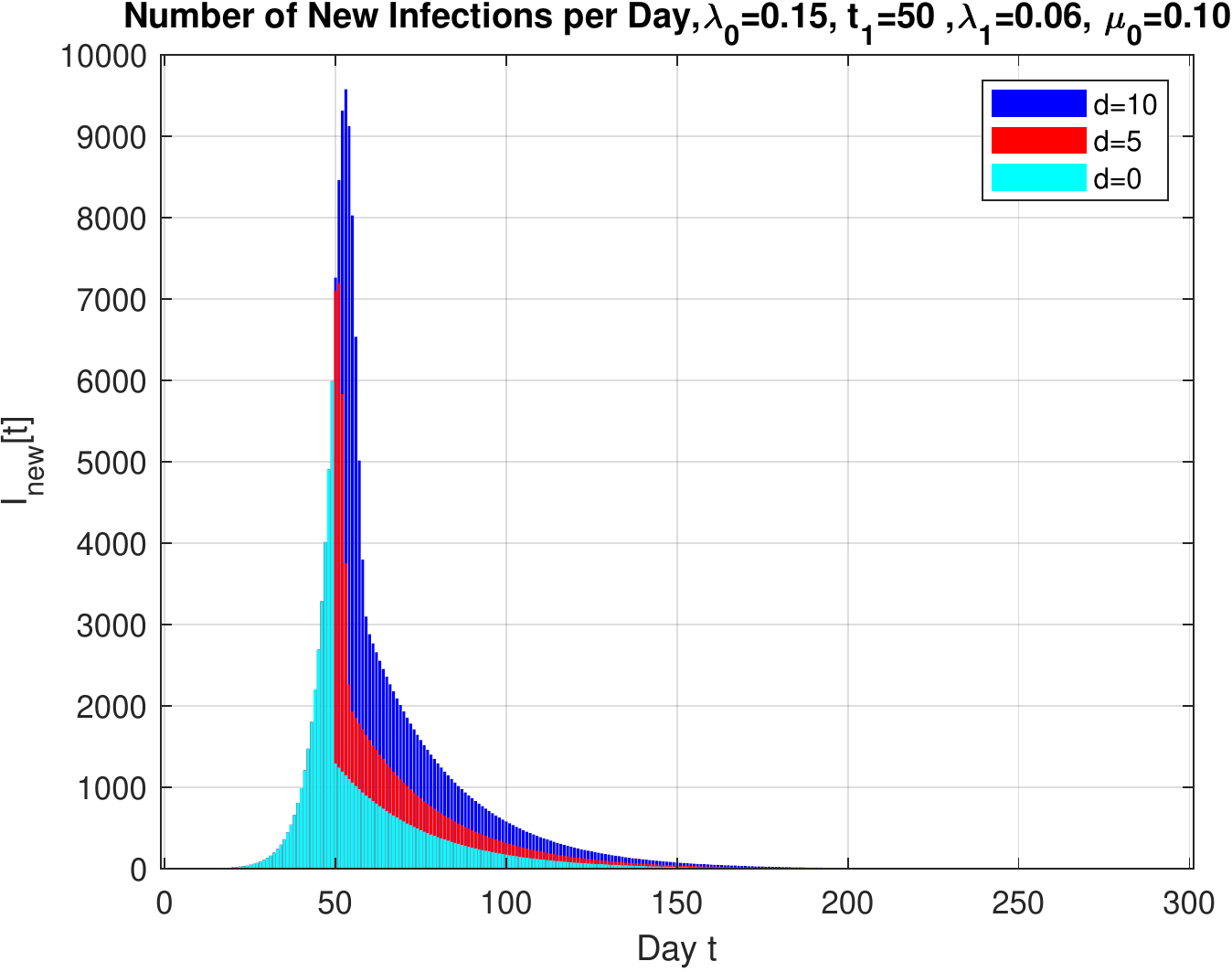}
\caption{\sf New Daily Infections $\oI_{new}[t], t=0, 1, 2, \ldots.$}
\label{fig:Example-new-infections-per-day}
\end{minipage}
\hspace{0.5cm}
\begin{minipage}[h]{0.45\textwidth}
\centering
\includegraphics[width=\textwidth]{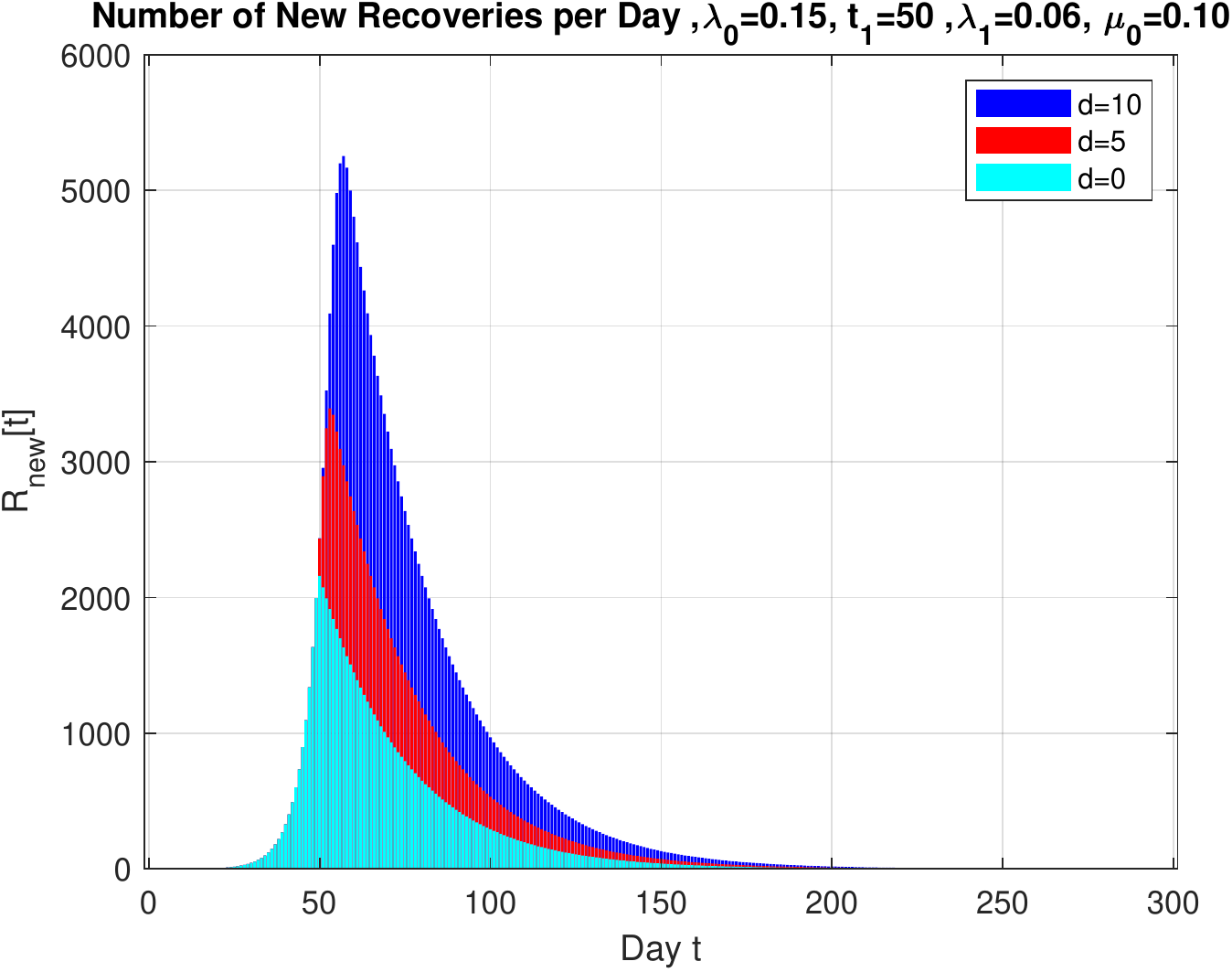} 
\caption{\sf New Daily Recoveries $\oR_{new}[t], t=0, 1, 2, \ldots.$ }
\label{fig:Example-new-recoveries-per-day}
\end{minipage} 
\end{figure}

\section{Time-Nonhomogeneous Stochastic Model}\label{sec:nonhomo-stochastoc-model}
\subsection{Probability Generating Function, Probability Mass Function, and Moments of $I(t)$}\label{subsec:find-PGF-PMF-moments}
The partial differential equation (PDE) that we defined in Part I, Section 3.1, Eqn.(15) can be generalized to the nonhomogeneous case as 
\begin{align}
\frac{\partial G(z,t)}{\partial t}
=(z-1)\left[(\lambda(t) z -\mu(t))\frac{\partial G(z,t)}{\partial z}+\nu(t) G(z,t)\right],\label{PDE-non-homo}
\end{align}
with the boundary condition
\begin{align}
G(z,0)=\Ex[z^{I(0)}]=z^{I_0}. \label{boundary-cond}
\end{align}  
Lagrange's method to solve (\ref{PDE-non-homo}) leads to the following auxiliary differential equations (see Part I, Appendix A, Eqn.(A.7)):
\begin{align}
\frac{dt}{1}= - \frac{dz}{(\lambda(t)z-\mu(t))(z-1)}=\frac{dG(z,t)}{\nu(t)(z-1)G(z,t)}. \label{auxiliary-eqn}
\end{align}
Unfortunately, the solution form given in (A.8) of Part I does not extend to the nonhomogeneous case. From the left and middle terms of (\ref{auxiliary-eqn}), we obtain
\begin{align}
\frac{dz}{dt}=(z-1)(\mu(t)-\lambda(t)z). \label{1st-aux-eqn}
\end{align}

If we change the variable $z$ to $x$, as Kendall \cite{kendall:1948a} suggests:
\begin{align}
x=\frac{1}{z-1},~~\mbox{or}~~z=1+\frac{1}{x},  
\end{align}
we find
\begin{align}
\frac{dz}{dt}=\frac{dz}{dx}\frac{dx}{dt}=-\frac{1}{x^2}\frac{dx}{dt},
\end{align}
which transforms (\ref{1st-aux-eqn}) into
\begin{align}
\frac{dx}{dt}=a(t)x+\lambda(t), ~~\mbox{where}~~a(t)=\lambda(t)-\mu(t),
\end{align}
which is an ordinary differential equation, similar to (\ref{diff-eq}). Thus, we have, similarly to (\ref{I(t)-and-C})
\begin{align}
xe^{-s(t)}= L(t) + C_1 \label{x-C_1}
\end{align}
where $s(t)$ is defined in (\ref{oI(t)}), and $L(t)$ is defined, similarly to $N(t)$ of (\ref{function-N(t)}), by
\begin{align}
L(t)\triangleq\int_0^t \lambda(u)e^{-s(u)}\,du. \label{function-L(t)}
\end{align}
 Then we have,
\begin{align}
\frac{e^{-s(t)}}{z-1}-L(t)=C_1.  \label{1st-solution}
\end{align}

Let us define the function $M(t)$ similar to the $L(t)$ above and $N(t)$ of (\ref{function-N(t)}):
\begin{align}
M(t)\triangleq\int_0^t \mu(u)e^{-s(u)}\,du.  \label{function-M(t)}
\end{align}
Then we can derive the following identity equation:
\begin{align}
L(t)-M(t)&=\int_0^t(\lambda(u)-\mu(u))e^{-s(u)}\,du=\int_0^t a(t)e^{-s(u)}\,du=\int_0^t s'(t)e^{-su)}\,du\nonumber\\
&= -\left[e^{-s(u)}\right]_{u=0}^t=1-e^{-s(t)}.\label{L-M-identity}
\end{align}

In order to find a second independent solution of (\ref{auxiliary-eqn}), 
we need to solve the differential equation
\begin{align}
\frac{d\log G(z,t)}{dz}=-\frac{\nu(t)}{\lambda(t) z-\mu(t)}, \label{2nd-auxiliary}
\end{align}
or alternatively
\begin{align}
\frac{d\log G(z,t)}{dt}=\nu(t)(z-1).\label{2nd-auxilary-2}
\end{align}
Unfortunately, neither of these equations seem unsolvable, unless $\nu(t)=0$.  Thus, we have to be content with the nonhomogeneous BD (birth-and-death) process, discussed by Kendall \cite{kendall:1948a}, and Bailey \cite{bailey:1964}.  

\subsection{Nonhomogeneous Birth-and-Death Process}\label{subsec:Nonhomo-BD-process}
Thus, we continue our analysis by assuming no external arrivals.  By setting the RHS of (\ref{2nd-auxiliary}) (hence  (\ref{2nd-auxilary-2}) as well) equal to zero, we find that the second solution to the  PDE (\ref{auxiliary-eqn}) is simply given by
\begin{align}
G(z,t)=C_2.
\end{align}
We write the functional relation between $C_1$ and $C_2$ as
\begin{align}
C_2=f(C_1),  \label{2nd-solution}
\end{align}
which, together with (\ref{1st-solution}), implies
\begin{align}
G(z,t)=f\left(\frac{e^{-s(t)}}{z-1}-L(t)\right).\label{G-and-f}
\end{align}

The boundary condition (\ref{boundary-cond}) gives
\begin{align}
z^{I_0}=f\left(\frac{1}{z-1}\right).\label{boundary-cond-z^I_0}
\end{align}
By setting
\begin{align}
\frac{1}{z-1}\triangleq y, ~~\mbox{or}~~z=1+\frac{1}{y},\label{y-and-z}
\end{align}
we find  the functional form $f(\cdot)$ , using (\ref{boundary-cond}), as
\begin{align}
f(y)=\left(1+\frac{1}{y}\right)^{I_0}.\label{function-f}
\end{align}
By combining  (\ref{G-and-f}) and (\ref{function-f}), we obtain 
\begin{align}
G(z,t)&=\left(1+ \frac{1}{\frac{e^{-s(t)}}{z-1}-L(t)}\right)^{I_0}.
\label{PGF_nonhomogeneous}
\end{align}
By substituting the following relation, which is from the identity equation (\ref{L-M-identity})
\begin{align}
e^{-s(t)}=1-L(t)+M(t), \label{E=1-L+M}
\end{align}
we find the that PGF of (\ref{PGF_nonhomogeneous}) can be written as
\begin{align}
G(z,t)&=\left(\frac{M(t)-(L(t)-1)z}{1+M(t)-L(t)z}\right)^{I_0},\label{PGF-M-L}
\end{align}
which we now express as
\begin{align}
G(z,t)=A(z,t)B(z,t),\label{G=AB}
\end{align}
where 
\begin{align}
A(z,t)&\triangleq (M(t)+(1-L(t))z)^{I_0}=\sum_{i=0}^{I_0}{I_0\choose i}M(t)^{I_0-i}(1-L(t))^iz^i,\label{A(t)}\\
B(z,t)&\triangleq (1+M(t)-L(t)z)^{-I_0}=(1+M(t))^{-I_0}\left(1-\frac{L(t)z}{M(t)+1}\right)^{-I_0}\nonumber\\
&=(1+M(t))^{-I_0}\sum_{j=0}^\infty {I_0+j-1\choose j}\left(\frac{L(t)}{M(t)+1}\right)^jz^j. \label{B(t)}
\end{align}
with the latter being obtained from the binomial theorem for negative integer exponents\footnote{
$(1-x)^{-n}=\sum_{j=0}^\infty {n+j-1\choose j}x^j, ~~|x|<1.$}

By summing all coefficients of the terms $z^k$ such that $k=i+j$ in (\ref{G=AB}), we find the expression for the probability mass functions $P_k(t)$:
\begin{align}
P_0(t)&=\left(\frac{M(t)}{1+M(t)}\right)^{I_0}\nonumber\\
P_k(t)&=\sum_{i=0}^{\min\{I_0,k\}}{I_0\choose i}{I_0+k-i-1\choose k-i}
\left(\frac{M(t)}{1+M(t)}\right)^{I_0-i}\left(\frac{L(t)}{1+M(t)}\right)^{k-i}\left(\frac{1-L(t)}{1+M(t)}\right)^i,~~\mbox{for}~~k\geq 1.  \label{P_k(t)-L-M}
\end{align}
By defining
\begin{align}
\alpha(t)\triangleq\frac{M(t)}{1+M(t)},~~\mbox{and}~~\beta(t)\triangleq\frac{L(t)}{1+M(t)}, \label{functions-alpha-beta}
\end{align}
we find 
\begin{align}
G(z,t)&=\left(\frac{\alpha(t)+(1-\alpha(t)-\beta(t))z}{1-\beta(t)z}\right)^{I_0},\label{G-alpha-beta}
\end{align}
which is a nonhomogeneous counterpart of the similar expression (see e.g.,  Takagi (see \cite{takagi:2007}, p. 92, also \cite{kobayashi:2020s,kobayashi:2020a})

The expression (\ref{P_k(t)-L-M}) can be written as 
\begin{align}
P_0(t)&=\alpha(t)^{I_0}\nonumber\\
P_k(t)&=\sum_{i=0}^{\min\{I_0,k\}}{I_0\choose i}{I_0+k-i-1\choose k-i}
\alpha(t)^{I_0-i}\beta(t)^{k-i}(1-\alpha(t)-\beta(t))^i,~~\mbox{for}~~k\geq 1.  \label{P_k(t)-alpha-beta}
\end{align}

When $I_0=1$, which is of our interest, the above formula can be simplified  to
\begin{align}
P_0(t)&=\alpha(t)\nonumber\\
P_k(t)&=(1-\alpha(t))(1-\beta(t)) \beta(t)^{k-1}.\label{P_k(t)-alpha-beta-I_0=1}
\end{align}
Using this, we can find an alternative way to compute the $\oI(t)$, viz.
\begin{align}
\oI(t)&=E[I(t)]=(1-\alpha(t))(1-\beta(t))\sum_{k=1}^\infty k\beta^{k-1}\nonumber\\
&=(1-\alpha(t))(1-\beta(t))\frac{d}{d\beta(t)}\left(\frac{1}{1-\beta(t)}\right)=\frac{1-\alpha(t)}{1-\beta(t)}.
\label{E[I(t)]-alpha-beta}
\end{align}

As a special case, let us consider the \emph{time-homogeneous} case, for which we have $s(t)=at=(\lambda-\mu)t$, and
\begin{align}
 L(t)&=\frac{\lambda}{a}(1-e^{-at}),
~~\mbox{and}~~ M(t)=\frac{\mu}{a}(1-e^{-at}).\label{M(t)-homogenous}
\end{align}
Thus, we have
\begin{align}
\alpha(t)&= \frac{\mu(e^{at}-1)}{\lambda e^{at}-\mu},
~~\mbox{and}~~\beta(t)=\frac{\lambda(e^{at}-1)}{\lambda e^{at}-\mu}=\frac{\lambda}{\mu}\alpha(t).
\label{alpha(t)-beta(t)-homogeneous}
\end{align} 
Then, we find that in the limit $t\to\infty$,
\begin{align}
\lim_{t\to\infty}\alpha(t)
 &=\left\{ \begin{array}{ll}\frac{\mu}{\lambda},~~&\mbox{if}~~a>0,\\
                             0,                 ~~&\mbox{if}~~a=0,\\
                             1,                 ~~&\mbox{if}~~a<0.
\end{array}\right. 
~~\mbox{and}~~
\lim_{t\to\infty}\beta(t)
 =\left\{ \begin{array}{ll}1                  ~~&\mbox{if}~~a>0,\\
                           0,                 ~~&\mbox{if}~~a=0,\\
                           \frac{\lambda}{\mu},~~&\mbox{if}~~a<0.
\end{array}\right. \label{limits-alpha-beta}
\end{align}
By substituting (\ref{alpha(t)-beta(t)-homogeneous})  into (\ref{G-alpha-beta}), we obtain
\begin{align}
G(z,t)&=\left(\frac{\mu(e^{at}-1)+(\lambda-\mu e^{at})z}{\lambda e^{at}-\mu-\lambda(e^{at}-1)z}\right)^{I_0},
\label{PGF-homogenous}
\end{align}
which is equivalent to what we found in Part I, Appendix A, Eqn. (A.19), with $r=0$.

\subsection{The Coefficient of Variation of the Nonhomogeneous BD Process}\label{subsec:Coeff-variation-nonhomo-BD-process}

In this section we calculate, from the PGF obtained above, the first and second moments of the infection process $I(t)$ for the nonhomogenous case.  As we did for the homogeneous case in Part I, let us take the natural logarithm of the PGF (\ref{PGF-M-L}), and differentiate it w.r.t. $z$, obtaining
\begin{align}
\frac{\partial \log G(z,t)}{\partial z}&=\frac{G'(z,t)}{G(z,t)}=I_0\left[\frac{1-L(t)}{M(t)-(L(t)-1)z}
+\frac{L(t)}{1+M(t)-L(t)z}\right].
\label{1st-deriv-G}
\end{align}
By setting $z=1$, we find
\begin{align}
E[I(t)]=\oI(t)=\frac{I_0}{1+M(t)-L(t)}=I_0 e^{s(t)},
\end{align}
which agrees with (\ref{I(t)-FA_process}) of the Feller-Arley process. 

In order to find the variance of $I(t)$, we differentiate  (\ref{1st-deriv-G}) once again, obtaining
\begin{align}
\frac{\partial^2 \log G(z,t)}{\partial z^2}=I_0\left[-\frac{(1-L(t))^2}{M(t)-((L(t)-1)z)^2} 
-\frac{L^2(t)}{(1+M(t)-L(t)z)^2}\right].
\end{align}
By setting $z=1$ in the above equation, we find the variance:
\begin{align}
\sigma_I^2(t)&=E[I^2(t)]-(E[I(t)])^2=\left.\frac{\partial^2 \log G(z,t)}{\partial z^2}\right|_{z=1}
+\left.\frac{\partial \log G(z,t)}{\partial z}\right|_{z=1} \nonumber\\
&=I_0e^{2s(t)}\left[L^2(t)-(1-L(t))^2\right]+ I_0e^{s(t)}
= I_0e^{2s(t)}(L(t)+M(t)).\label{variance-nonhomo-I(t)}
\end{align}
Thus, the coefficient of variation of $I(t)$, $c_I(t)$ is given by
\begin{align}
c_I(t)=\frac{\sigma_I(t)}{\oI(t)}=\sqrt{\frac{L(t)+M(t)}{I_0}}.\label{cv-nonhomo-I(t)}
\end{align} 

%\clearpage

\section{Numerical Analysis of Time-Nonhomogeneous Stochastic Model}
\label{subsec:numerical-example-stochastic-nonhomo}

Let us continue pursuing the numerical example discussed in Section \ref{sec:Example-nonhomogenous-deterministic}, where $\lambda(t)$ begins to decrease from $\lambda_0$ to $\lambda_1$ at $t_1=50$, as depicted in Figure \ref{fig:lambda(t)-raised-cosine}.  

\subsection{Functions $\Sigma^{-1}(t), \alpha(t)$, and the $c_I(t)$}\label{subsubsec:Sigma-alpha-c_I(t)} 
\begin{figure}[thb]
\begin{minipage}[t]{0.30\textwidth}
\centering
\includegraphics[width=\textwidth]{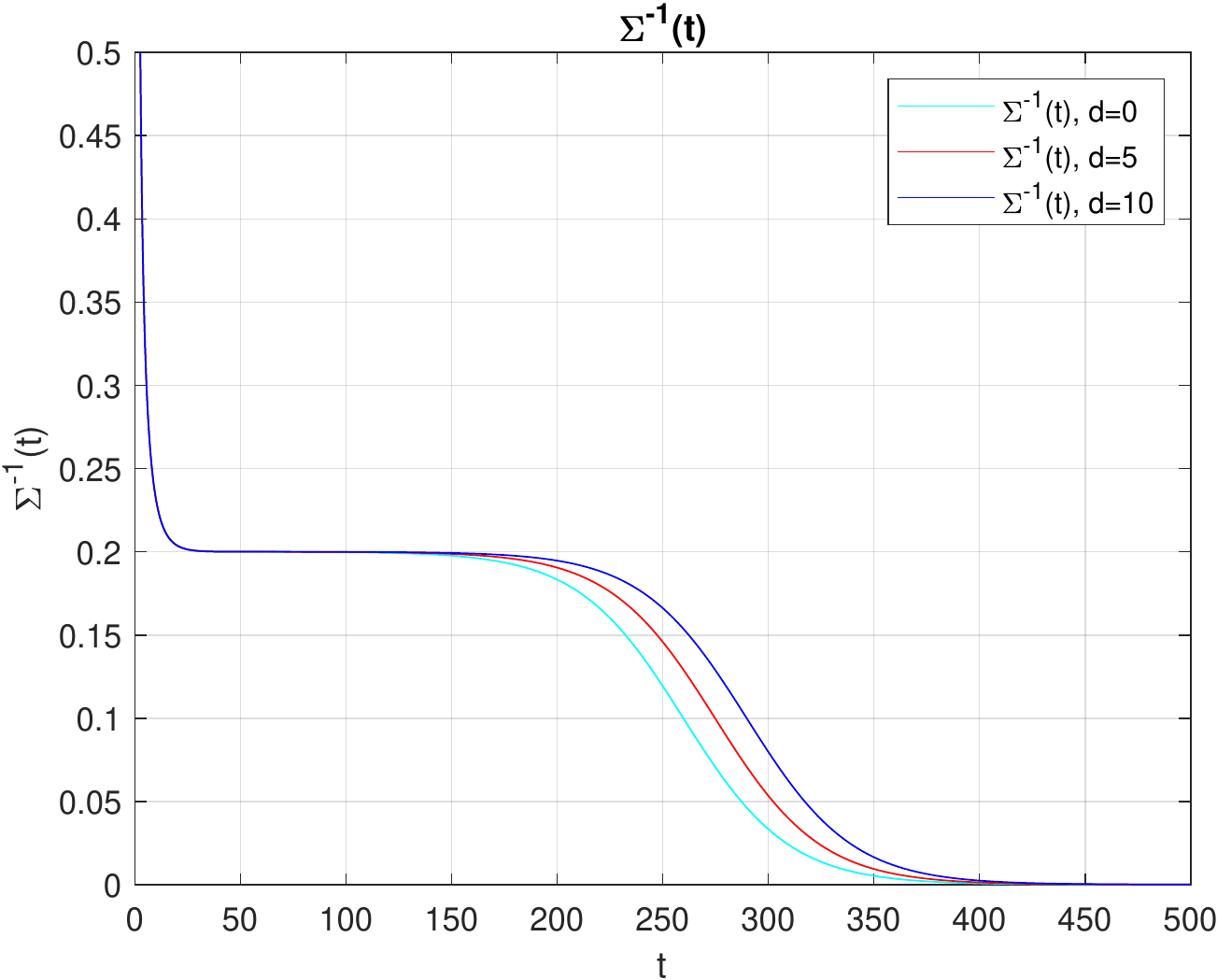}
\caption{\sf $\Sigma^{-1}(t)$, for $0\leq t\leq 500$.}
\label{fig:Sigma-1(t)}
\end{minipage}
\hspace{0.5cm}
\begin{minipage}[t]{0.30\textwidth}
\centering
\includegraphics[width=\textwidth]{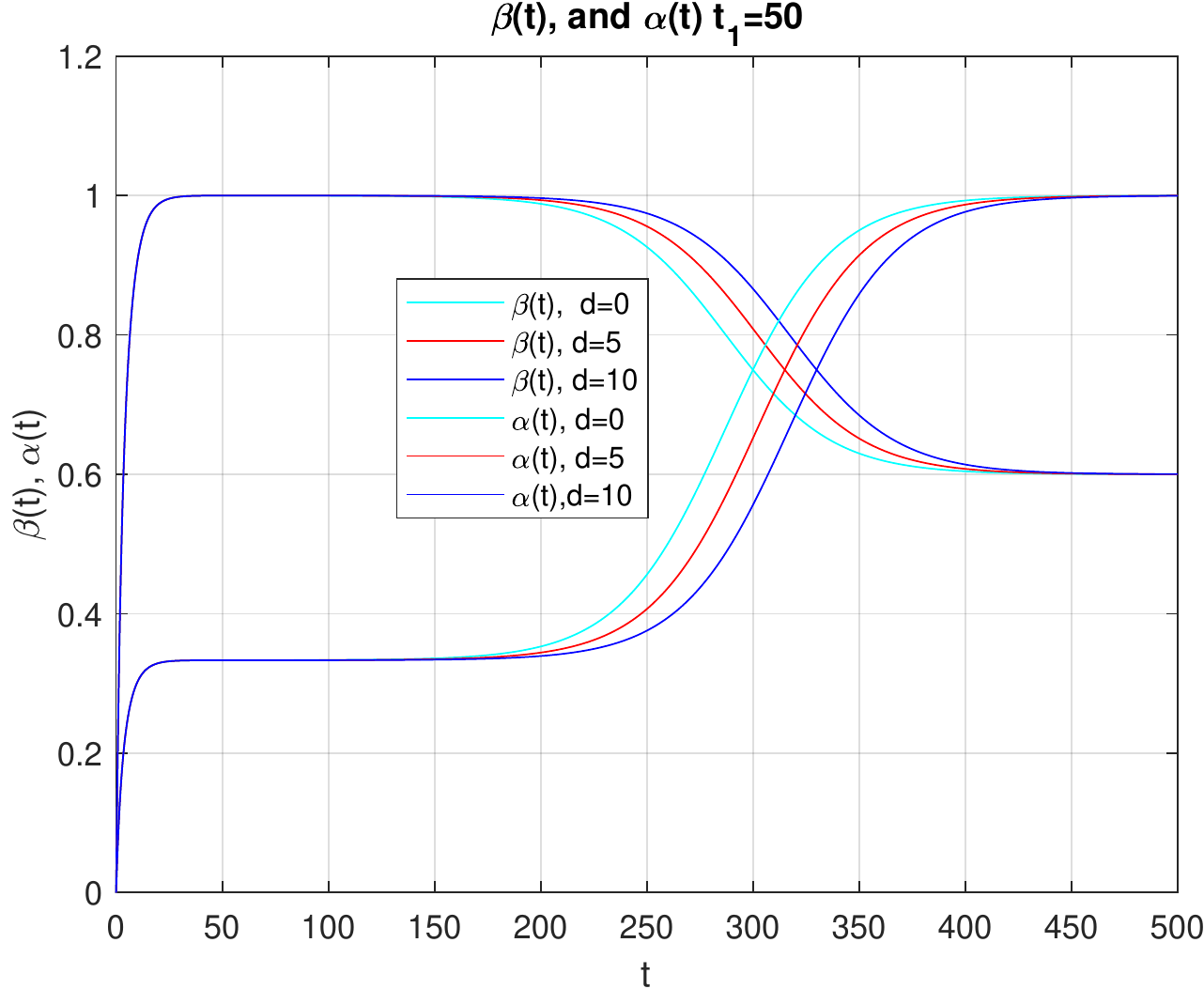}
\caption{\sf $\alpha(t)$ \& $\beta(t)$ for $0\leq t\leq 500$.}
\label{fig:alpha(t)-and-beta(t)}
\end{minipage}
\hspace{0.5cm}
\begin{minipage}[t]{0.30\textwidth}
\centering
\includegraphics[width=\textwidth]{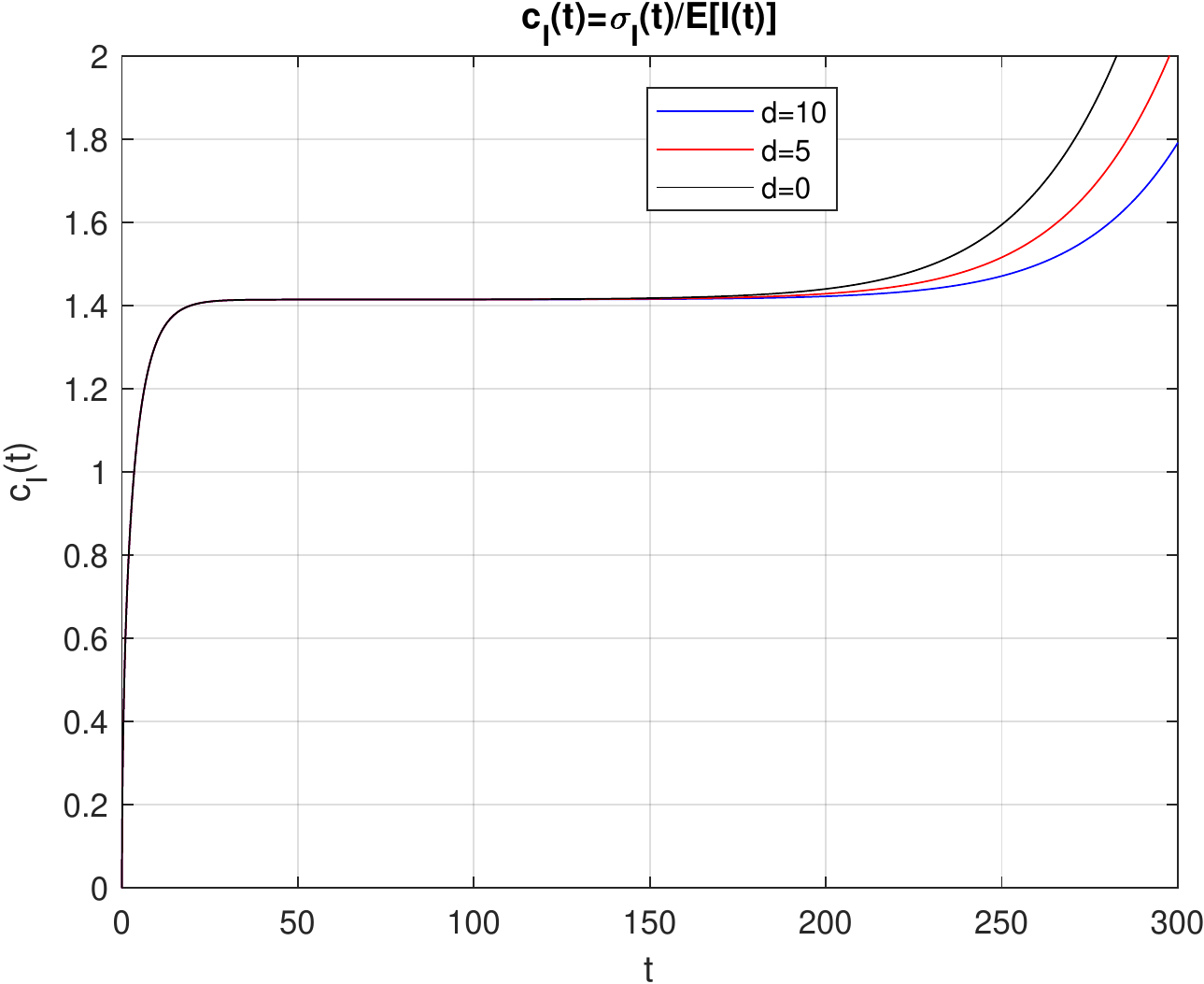}
\caption{\sf The coefficient of variation of $I(t)$, where $0<t<300$}
\label{fig:c_I(t)}
\end{minipage}
\end{figure}

In the present subsection, we present three most important graphical curves which characterize important behavior of the stochastic process $I(t)$.

The first is the behavior of the function $\Sigma^{-1}(t)\triangleq 1/\Sigma(t)$, where $\Sigma(t)$ is an integration of $e^{-s(t)}$ as defined in (\ref{function-Sigma(t)}), and plotted in Figure \ref{fig:Sigma(t)}.  The shape of $\Sigma(t)$ is determined by the behavior of the function $s(t)$ at around $s(t)\approx 0$ as discussed in Section \ref{subsec:derivation of oI(t)}.  The function $\Sigma(t)$, in turn, largely determines the shapes of the functions $N(t), L(t)$ and $M(t)$, defined by (\ref{function-N(t)}), (\ref{function-L(t)}) and (\ref{function-M(t)}), respectively.  In Figure \ref{fig:Sigma-1(t)}, we plot the function $\Sigma^{-1}(t)(=1/\Sigma(t))$. 

In Figure \ref{fig:alpha(t)-and-beta(t)}, we show the functions $\alpha(t)$ and $\beta(t)$ defined by (\ref{functions-alpha-beta}).  The function $\alpha(t)$ rapidly rises to $\frac{\mu_0}{\lambda_0}=1/3$,  as seen from (\ref{limits-alpha-beta}), stays at this level for a long time, and then begins at $t\approx 200$  to rise towards ``1", which is the limit of $\alpha(t)$ in the regime $a<0$ .  From (\ref{P_k(t)-alpha-beta-I_0=1}) we see that when $I_0=1$, the function $\alpha(t)$ is equal to $P_0(t)$.  

As the second equation in (\ref{limits-alpha-beta}) shows, the function $\beta(t)$ quickly reaches the unit level in the regime $a>0$, and stays there until $t\approx 200$, and then makes a transition to the limit $\frac{\lambda_1}{\mu_1}=0.6$ in the regime $a<0$. Note that the transitions of $\alpha(t)$ and $\beta(t)$ from their limit values in the first regime $a>0$ to the limit values in the regime $a<0$ require some time,  taking as long as 200 days (i.e., from $t\approx 200$ to $t\approx 400$) in our running example.  The horizontal axis of this Figure \ref{fig:alpha(t)-and-beta(t)} is extended up to $t=500$, whereas in most other figures of this example we plot only up to $t=300$. 

This transitional behavior of the functions $\alpha(t)$, together with the large coefficient of variation, is  perhaps the most important result of the present article, that is,

\begin{itemize}

\item[(i)] Although the expected infection function $\oI(t)$ begins to decrease once $\lambda(t)$ becomes smaller than $\mu(t)$ (hence $a(t)<0$), $\oI(t)$ does not become sufficiently small until the function $s(t)=\int_0^t a(u)\,du$ 
decreases near zero (see Figure \ref{fig:s(t)-bent-at-t=50}). In this running example, it is not until $t\approx t_1+250=300$. This can be easily computed from Figure \ref{fig:s(t)-bent-at-t=50}, from the value of $s(t_1)(\approx 10)$ and the slope $a(t)(\approx -0.04)$,i.e.,  $s(t_1)/|a(t)|=250$.  An important observation to make is that there is a considerable delay from the moment $t_1(=50)$ when a state of emergency was issued by the government until the time the infected population decreases towards zero. In our running example, it takes as many as 150 (=200-50) days until $I(t)$ becomes very small, on average, and additional 200 (=400-200) days until the infection process is expected to come to a full end. 

\item[(ii)] The transitional behavior of $\alpha(t)$ can be better understood by rewriting $\alpha(t)$ as
\begin{align}
\alpha(t)=\frac{1}{1+\mu^{-1}\Sigma^{-1}(t)},
\end{align}
As we noted earlier, $P_0(t)$, the probability that the number of infected persons at time $t$ is zero, is equal to $\alpha(t)^{I_0}$ (c.f. (\ref{P_k(t)-alpha-beta})).  Thus, the function $\alpha(t)$ should provide the policy makers with such crucial a guidance as how long its state of emergency declaration should be maintained.

\item[(iii)] Shown in Figure \ref{fig:c_I(t)} is the coefficient of variation (CV) $c_I(t)$ of the process $I(t)$, as given in  (\ref{cv-nonhomo-I(t)}). Note that the CV remains nearly constant ($\approx 1.44$ in this example) until $t\approx 200$, when the CV begins to rise. The value of the CV as large as 1.44 is due to the fact the PMF $P_k(t)$ at any given $t$ is nearly flat except its value at $k=0$ (see Figures \ref{fig:Non-homo-BD-P_k(0)-I_0=1} through \ref{fig:Non-homo-BD-P_k(350)-I_0=1} in the next section). Such a highly skewed distribution with an extremely long tail implies that the $\oI(t)$ obtained from any deterministic model can provide very limited information of the stochastic process $I(t)$. In the same token, any parameter estimations obtained from the observed data of an instance of $I(t)$ will be very unreliable, to the say the least: an estimated value of the model parameter  may be significantly deviated, more often than not, from its true (unknown) value. We will defer a full discussion on model parameter estimation until Part IV \cite{kobayashi:2021c}.

It may look counter-intuitive that the CV $c_I(t)$ begins to increase exponentially after $t$ goes beyond $\approx 200$, considering that $\oI(t)$ becomes practically down towards zero in this time frame.  This somewhat surprising behavior can be explained by observing (\ref{variance-nonhomo-I(t)}):  the variance $\sigma_I(t)$ is nearly proportional to $\sqrt{\Sigma(t)}$. 

\end{itemize}
\subsection{Numerical Plots of the Time-Dependent PMF $P_k(t)$} \label{subsec:PMF}

We obtained in Section \ref{subsec:Nonhomo-BD-process} a closed-form expression for $P_k(t)=\Pr[I(t)=k]$, which is given  by (\ref{P_k(t)-alpha-beta}).  If the initial value $I_0$ is 1, (\ref{P_k(t)-alpha-beta}) reduces to (\ref{P_k(t)-alpha-beta-I_0=1}).

From the assumption $\Pr[I(0)=I_0]=1$, it should follow that $P_k(0)=\delta_{k,I_0}$, where $\delta_{i, j}$ is Kronecker's delta \footnote{
\begin{align*}
\delta_{i,j}=\left\{\begin{array}{ll}&1,~~\mbox{if}~~i=j\\                                                                 &0,~~\mbox{otherwise}\end{array}\right.
\end{align*}
}.
To verify this, consider
\begin{align}
P_{I_0}(t)&=
\sum_{i=0}^{I_0}{I_0\choose i}{2I_0-i-1\choose I_0-i}\alpha(t)^{I_0-i}\beta(t)^{I_0-i}(1-\alpha(t)-\beta(t))^i
\end{align}
As $t\to 0$, $\alpha(t)\to 0$ and $\beta(t)\to 0$.  Then, all the terms in the above expression approach zero, except for the term that corresponds to $i=I_0$, i.e.,
\begin{align}
\lim_{t\to 0}P_{I_0}(t)&={I_0\choose I_0}{I_0-1\choose 0}\lim_{t\to 0}\alpha(t)^0\beta(t)^0(1-\alpha(t)-\beta(t))^{I_0}
=1^{I_0}=1.
\end{align}

In the limit $t\to \infty$, we have 
\begin{align}
\lim_{t\to\infty}P_k(t)=\lim_{t\to\infty}\alpha(t)^{I_0}=1, \label{extinction}
\end{align}
which implies that the infection $I(t)$ converges to zero with probability one as $t\to\infty$.  This is not surprising in view of the fact that we assume no external infected individuals (i.e., no immigration) and that we assume the effective reproduction number ${\cal R}(t)\triangleq \lambda(t)/\mu(t)<1$ for $t>t_1+d$.  In the population study, the phenomenon (\ref{extinction}) is called the \emph{extinction} of species under study.

Presented below are various cross-sections of the three-dimensional array 
\begin{align}
Z(t, k)\triangleq P_k(t),~~t\geq 0, k=0, 1, 2, \ldots, \label{Z(t,k)}
\end{align}
for the running example discussed in this article.  The first group of the plots are the PMFs of $P_k(t)$ at various points in time $t$. 

The second group is a set of plots of $P_k(t)$ vs. $t$ for a given $k=0, 1, 2, \ldots$. Note that this is not any sort of probability distribution function.

We then finally present the surface plot of the function $Z(t,k)$ of (\ref{Z(t,k)}) to provide an overall picture of the time-dependent PMF $P_k(t)$.   

The whole purpose of presenting these graphical plots is to show that the range of values that infinitely many possible realizations  (i.e., sample paths) of the BD process $I(t)$ may take on is so broad.  Thus, the expected value $\oI(t)$ or a particular instance or sample path of the stochastic process $I(t)$ \emph{cannot be claimed as a typical instance}.  This implies that we must be extremely careful in drawing any sort of conclusion on the property of $I(t)$ from a small number of realization of the process.  This point will be made further clearer in the companion paper \cite{kobayashi:2021b}, where we will present extensive results on non-homogeneous BD and BDI processes. 

\begin{enumerate}

\item \textbf{Time Dependent PMF $\mathbf{P_k(t)}$ at various instants $\textbf{t}$, with $I_0=1$}

\begin{figure}[thb]
\begin{minipage}[t]{0.30\textwidth}
\centering
\includegraphics[width=\textwidth]{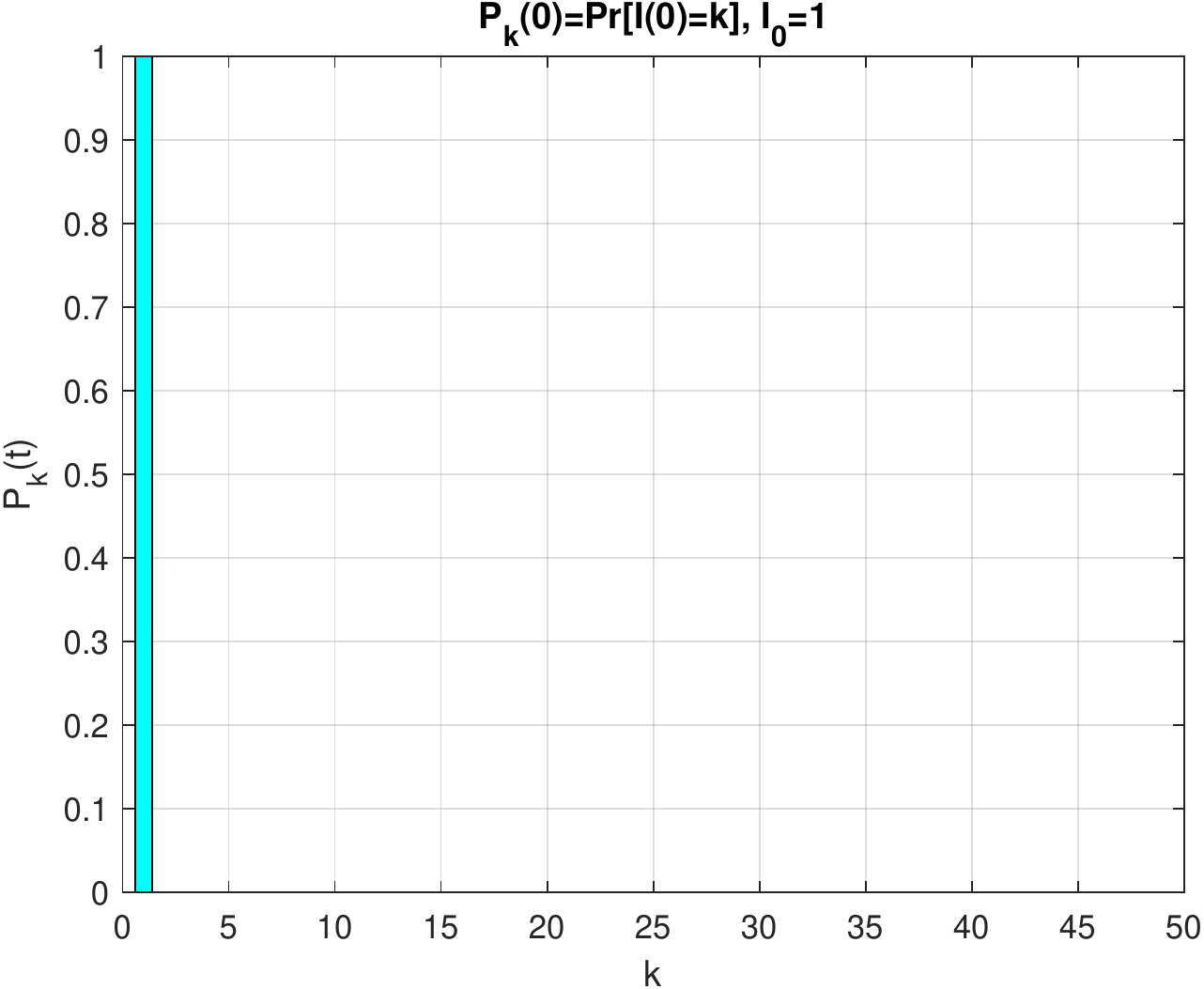}
\caption{\sf $P_k(0)$.}
\label{fig:Non-homo-BD-P_k(0)-I_0=1}
\end{minipage}
\hspace{0.5cm}
\begin{minipage}[t]{0.30\textwidth}
\centering
\includegraphics[width=\textwidth]{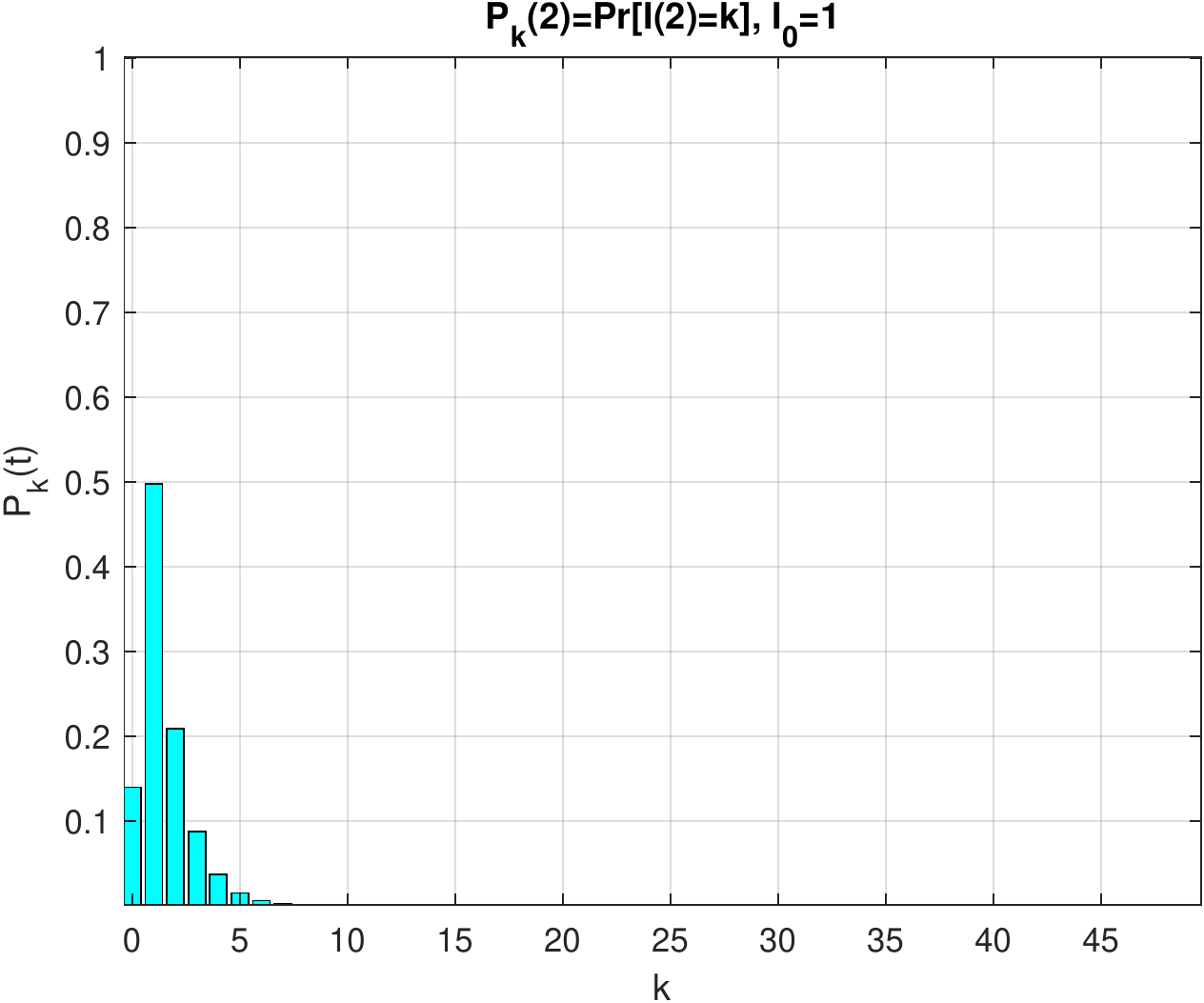}
\caption{\sf $P_k(2)$.}
\label{fig:Non-homo-BD-P_k(2)-I_0=1}
\end{minipage}
\hspace{0.5cm}
\begin{minipage}[t]{0.30\textwidth}
\centering
\includegraphics[width=\textwidth]{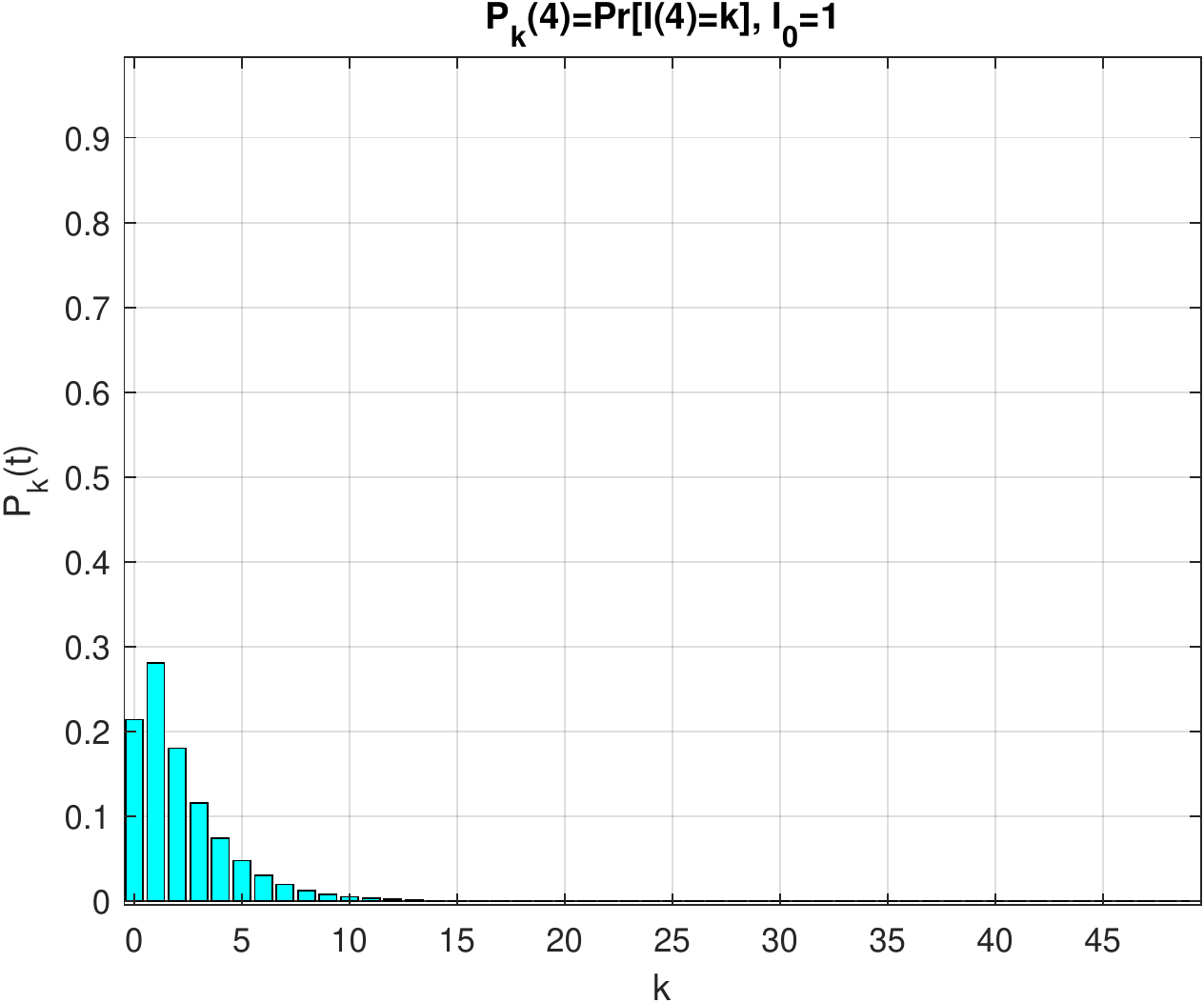}
\caption{\sf $P_k(4)$.}
\label{fig:Non-homo-BD-P_k(4)-I_0=1}
\end{minipage}
\end{figure}
\begin{figure}[hbt]
\begin{minipage}[h]{0.30\textwidth}
\centering
\includegraphics[width=\textwidth]{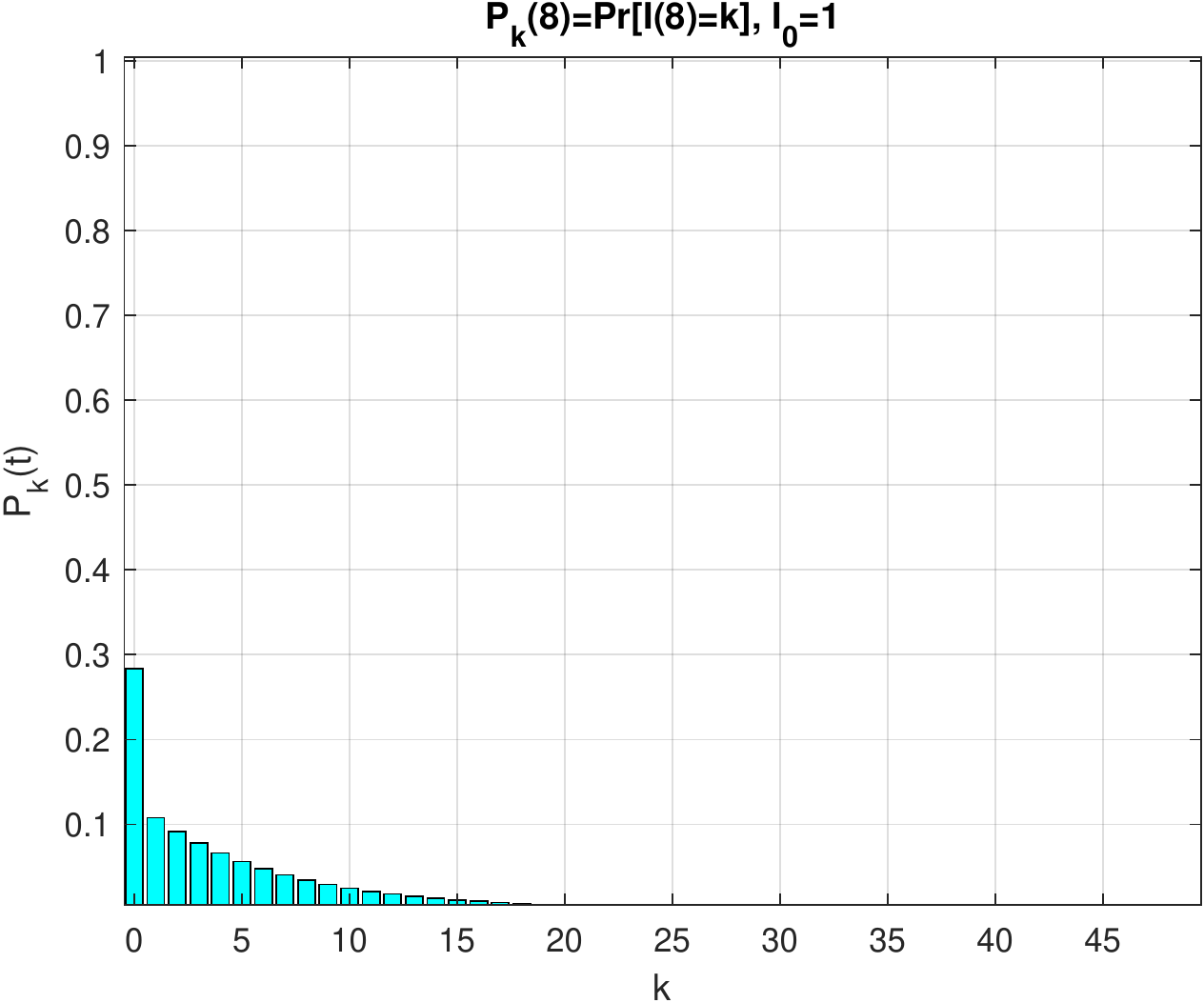}
\caption{\sf  $P_k(8)$.}
\label{fig:Non-homo-BD-P_k(8)-I_0=1}
\end{minipage}
\hspace{0.5cm}
\begin{minipage}[h]{0.30\textwidth}
\centering
\includegraphics[width=\textwidth]{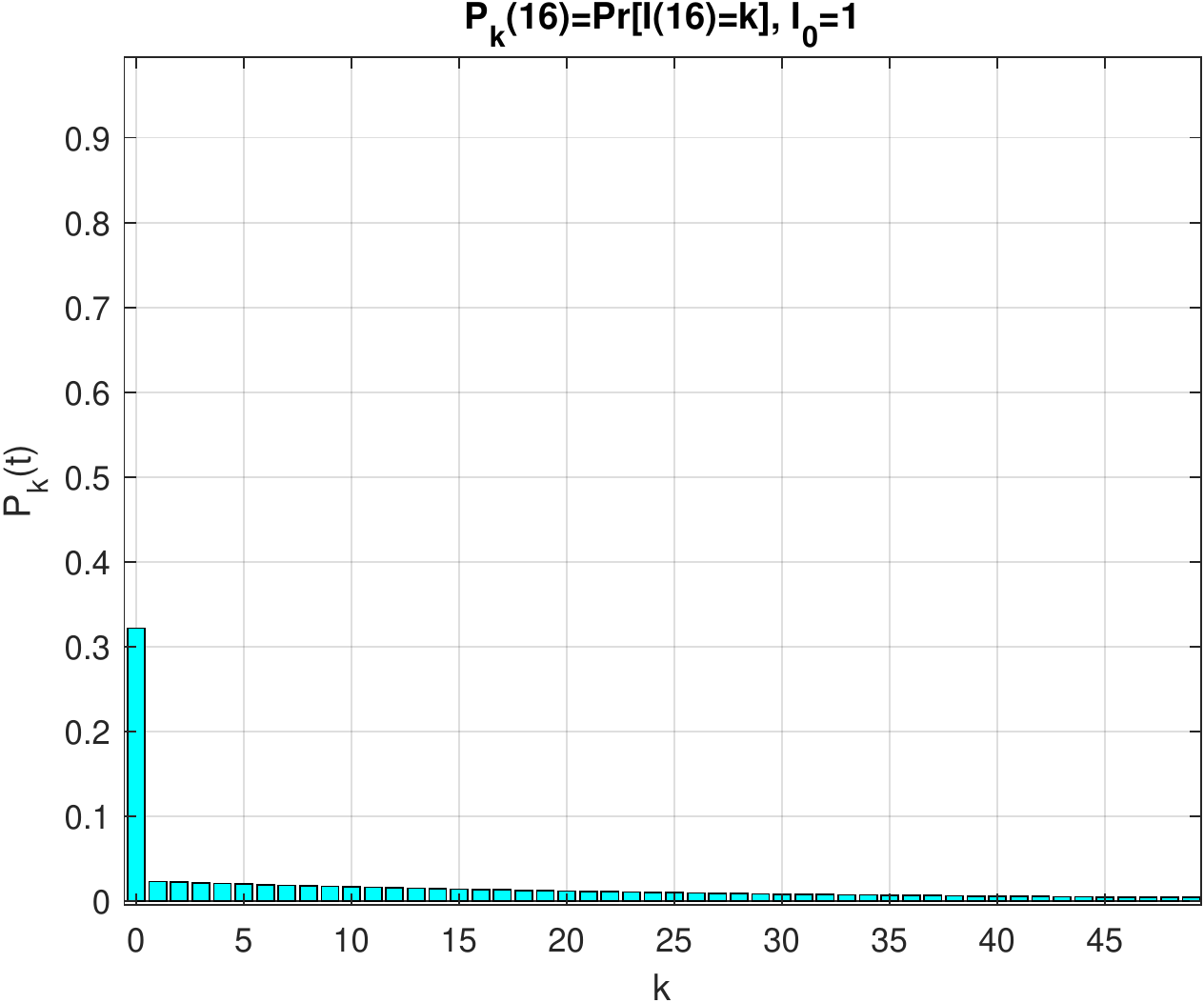}
\caption{\sf $P_k(16)$.}
\label{fig:Non-homo-BD-P_k(16)-I_0=1}
\end{minipage}
\hspace{0.5cm}
\begin{minipage}[h]{0.30\textwidth}
\centering
\includegraphics[width=\textwidth]{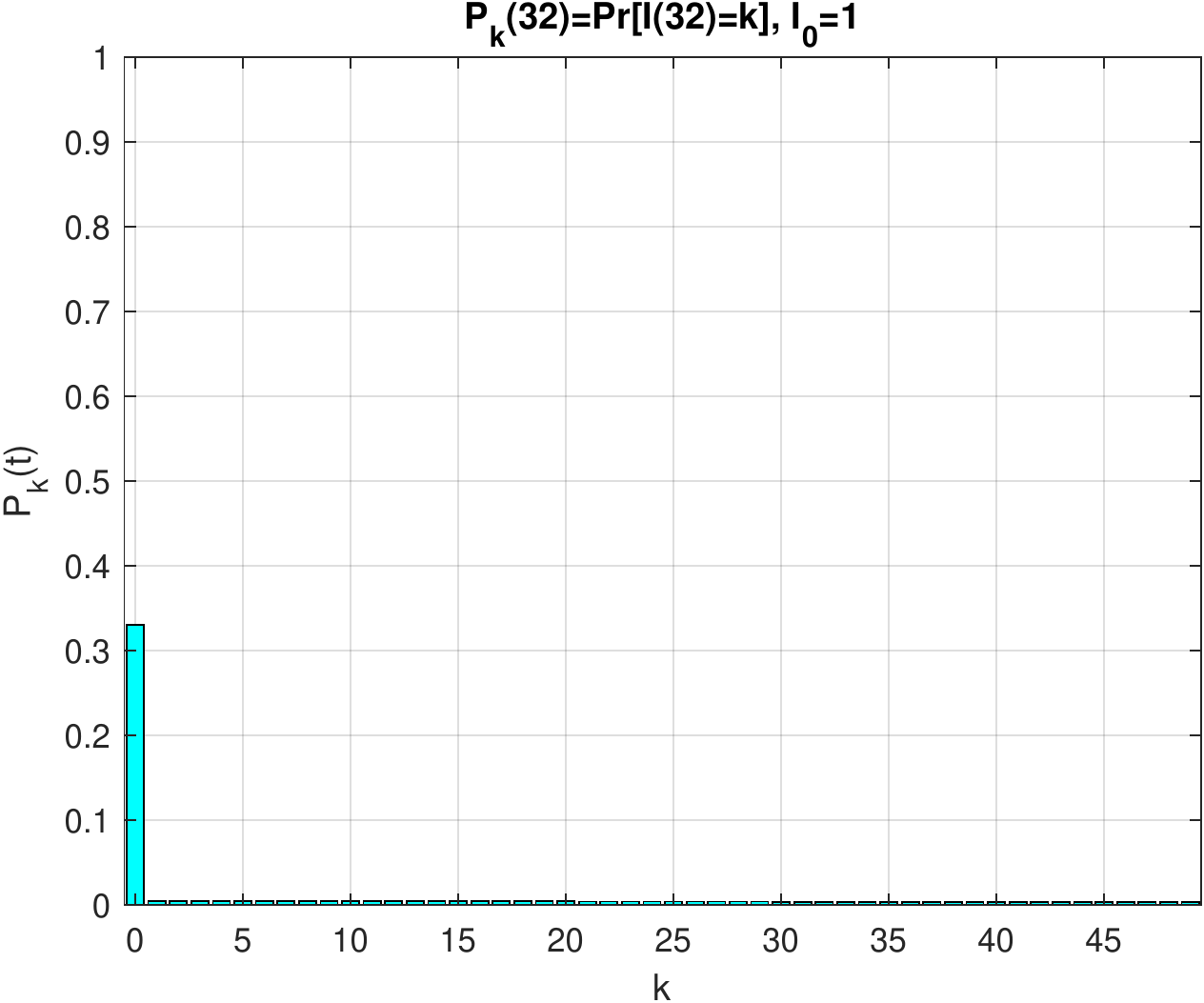}
\caption{\sf $P_k(32)$.}
\label{fig:Non-homo-BD-P_k(32)-I_0=1}
\end{minipage}
\end{figure}
\begin{figure}[bth]
\begin{minipage}[b]{0.30\textwidth}
\centering
\includegraphics[width=\textwidth]{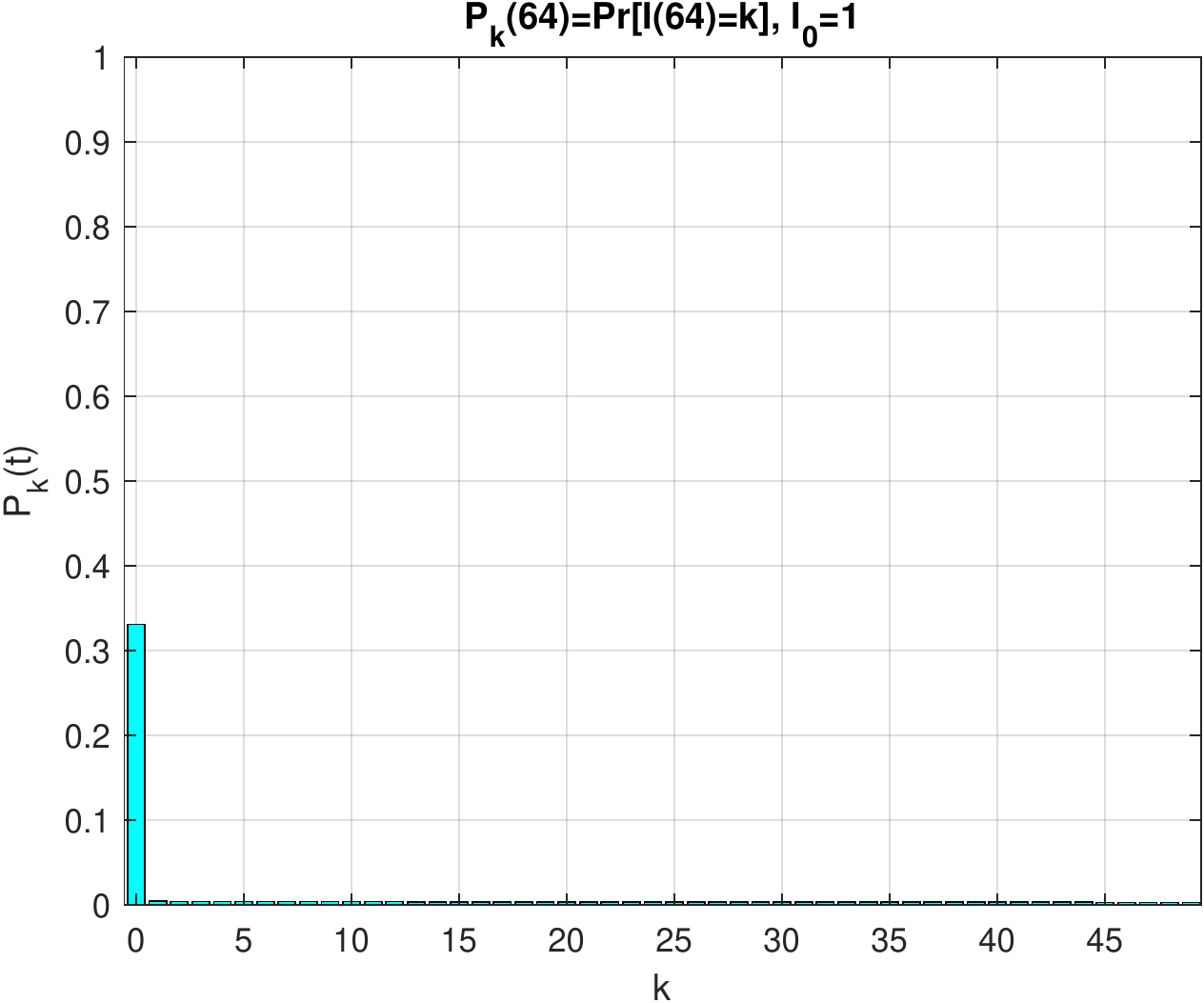}
\caption{\sf $P_k(64)$.}
\label{fig:Non-homo-BD-P_k(64)-I_0=1}
\end{minipage}
\hspace{0.5cm}
\begin{minipage}[b]{0.30\textwidth}
\centering
\includegraphics[width=\textwidth]{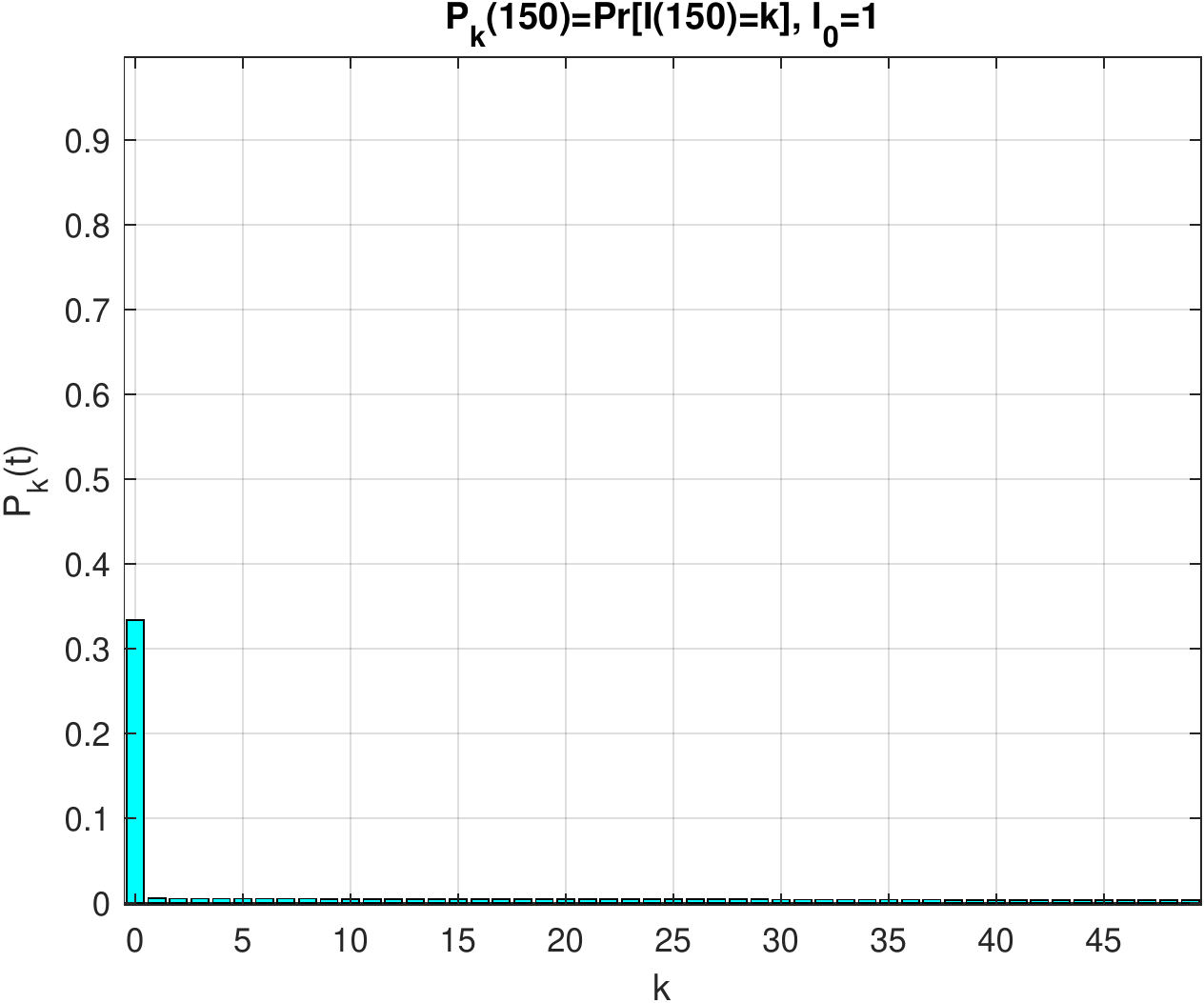}
\caption{\sf $P_k(150)$.}
\label{fig:Non-homo-BD-P_k(150)-I_0=1}
\end{minipage}
\hspace{0.5cm}
\begin{minipage}[b]{0.30\textwidth}
\centering
\includegraphics[width=\textwidth]{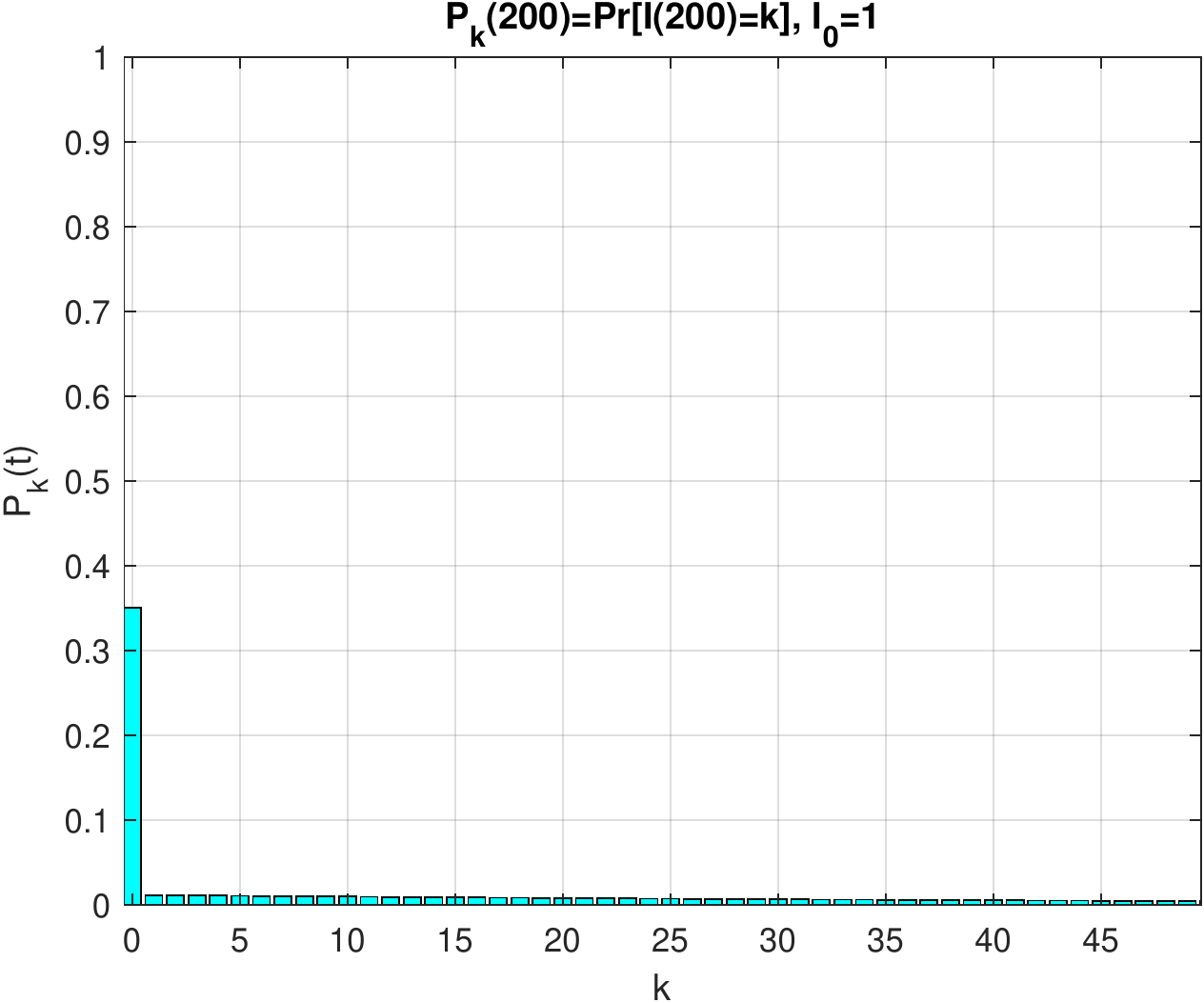}
\caption{\sf $P_k(200)$.}
\label{fig:Non-homo-BD-P_k(200)-I_0=1}
\end{minipage}
\end{figure}

\begin{figure}[thb]
\begin{minipage}[t]{0.30\textwidth}
\centering
\includegraphics[width=\textwidth]{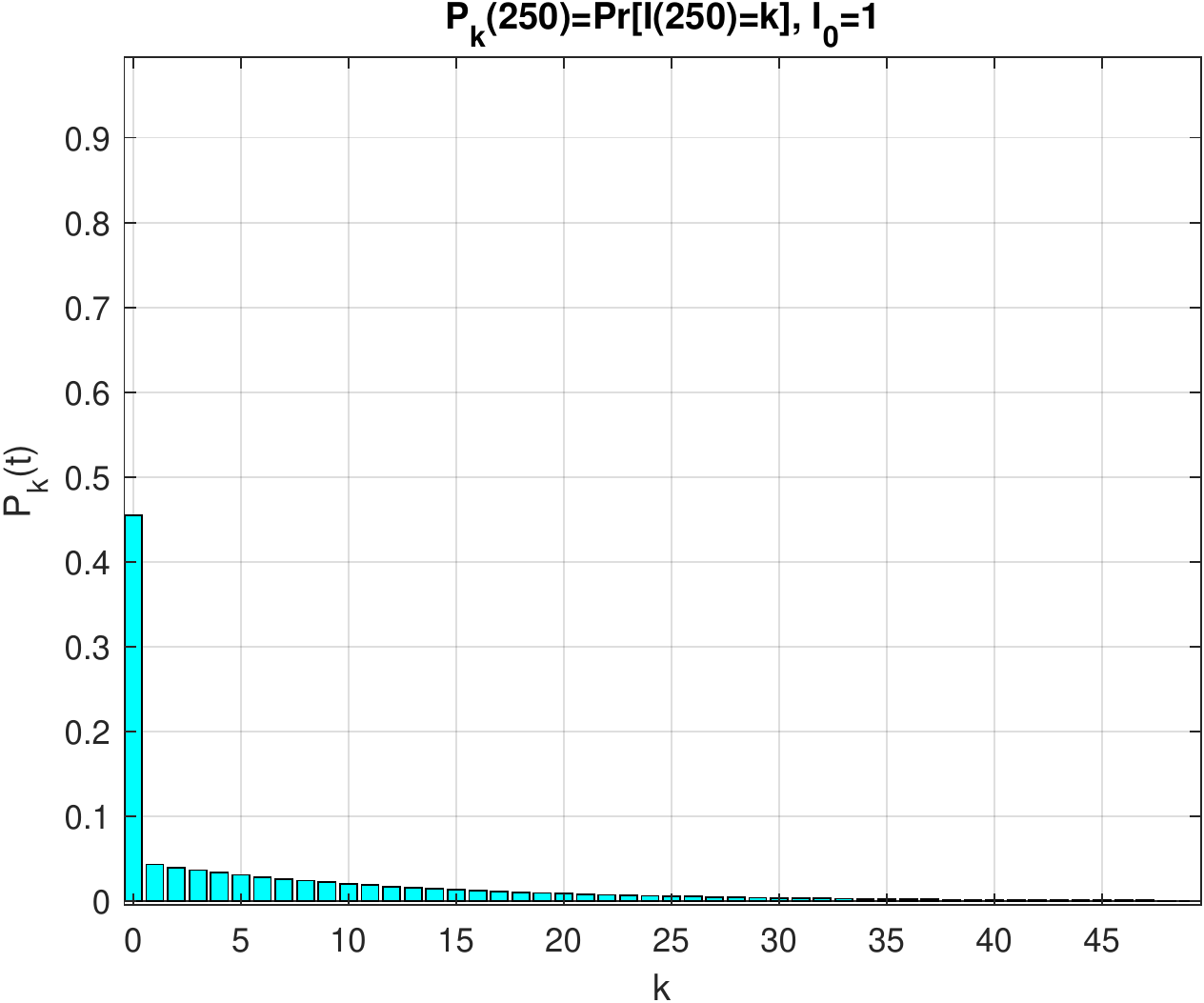}
\caption{\sf $P_k(250)$.}
\label{fig:Non-homo-BD-P_k(250)-I_0=1}
\end{minipage}
\hspace{0.5cm}
\begin{minipage}[t]{0.30\textwidth}
\centering
\includegraphics[width=\textwidth]{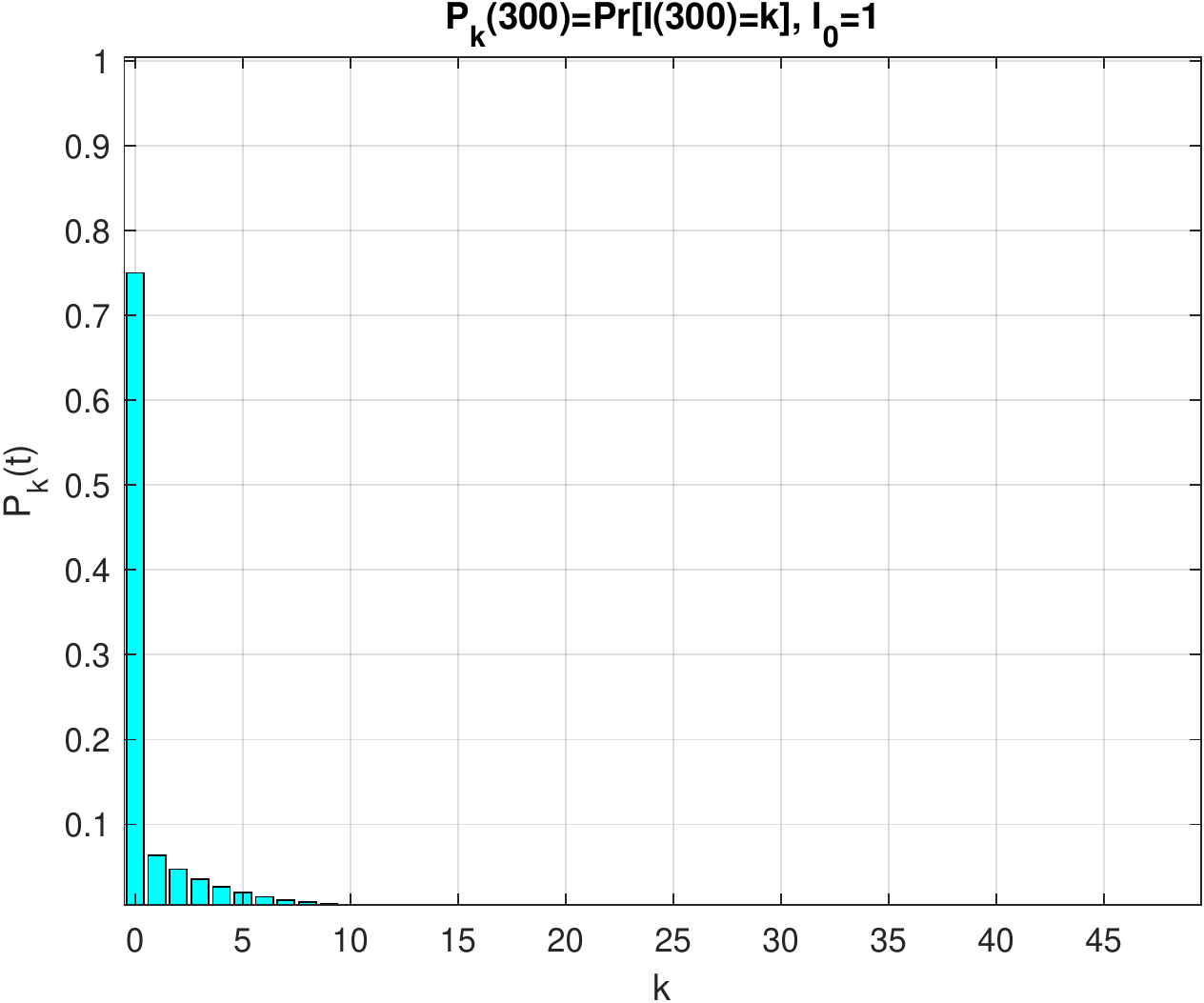}
\caption{\sf $P_k(300)$.}
\label{fig:Non-homo-BD-P_k(300)-I_0=1}
\end{minipage}
\hspace{0.5cm}
\begin{minipage}[t]{0.30\textwidth}
\centering
\includegraphics[width=\textwidth]{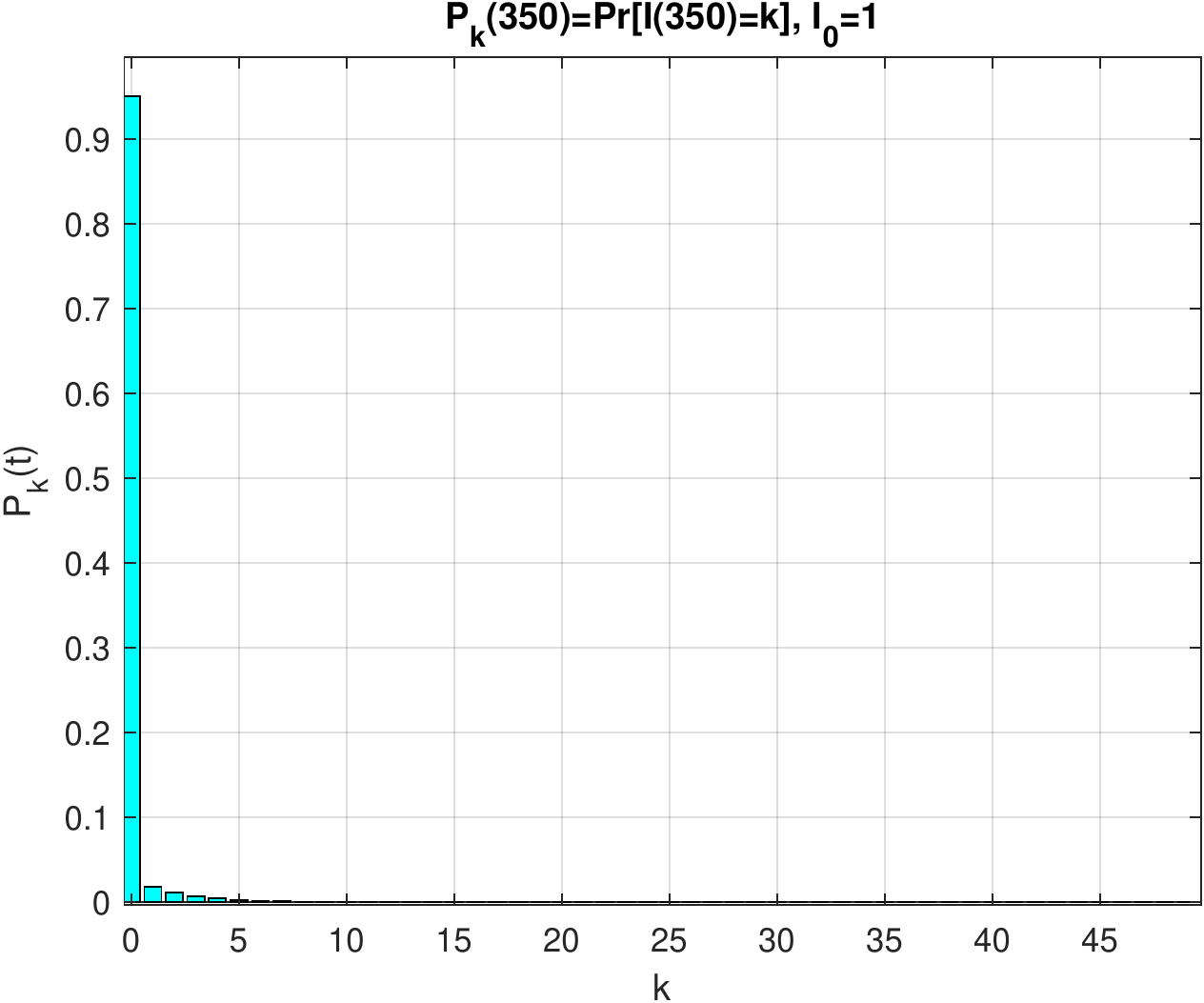}
\caption{\sf $P_k(350)$.}
\label{fig:Non-homo-BD-P_k(350)-I_0=1}
\end{minipage}
\end{figure}

\clearpage

\item \textbf{$P_k(t)$ vs. $t$ for given $k$}:

\begin{figure}[thb]
\begin{minipage}[t]{0.30\textwidth}
\centering
\includegraphics[width=\textwidth]{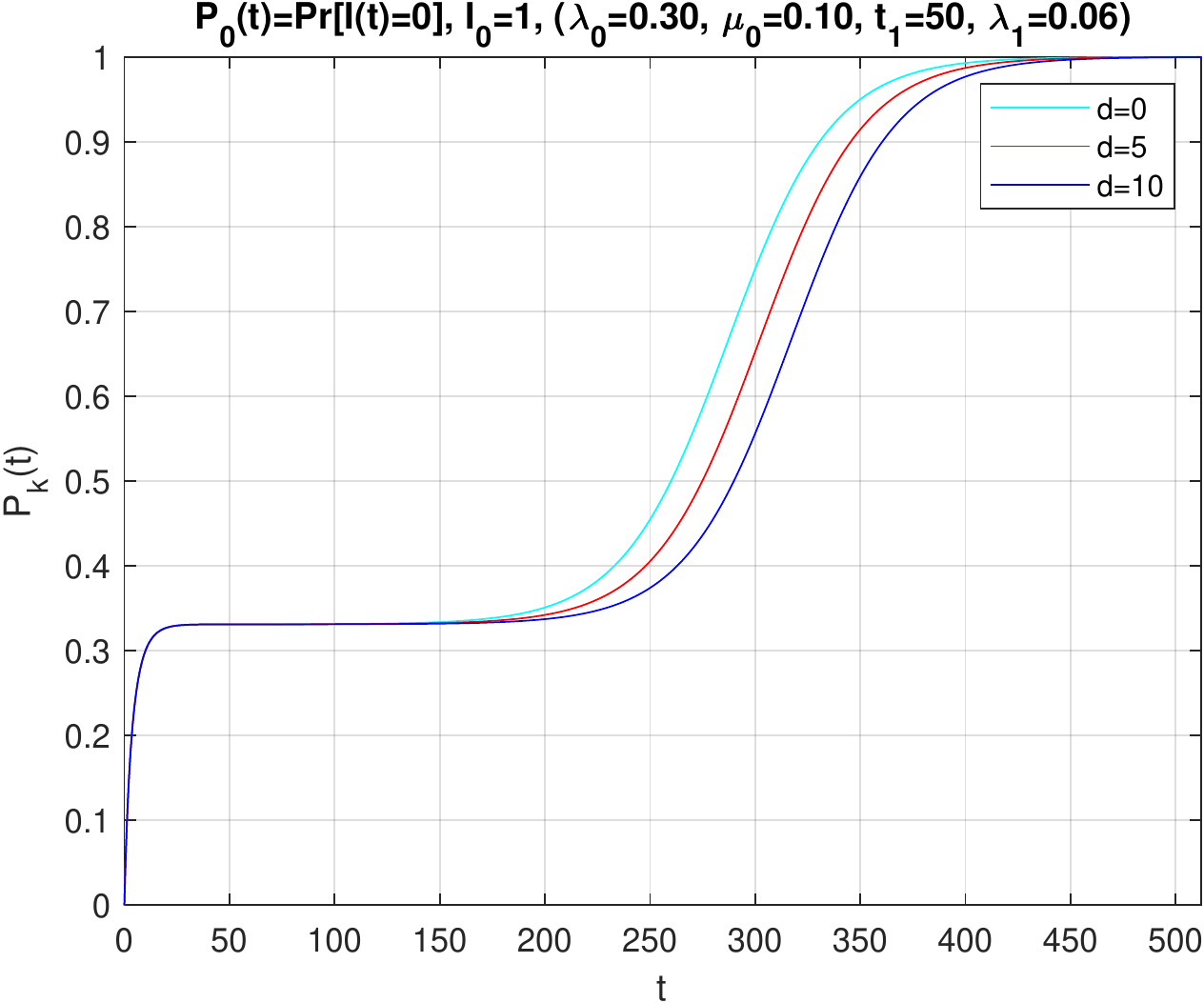}
\caption{\sf $P_0(t)$.}
\label{fig:Non-homo-BD-P_0(t)-I_0=1}
\end{minipage}
\hspace{0.5cm}
\begin{minipage}[t]{0.30\textwidth}
\centering
\includegraphics[width=\textwidth]{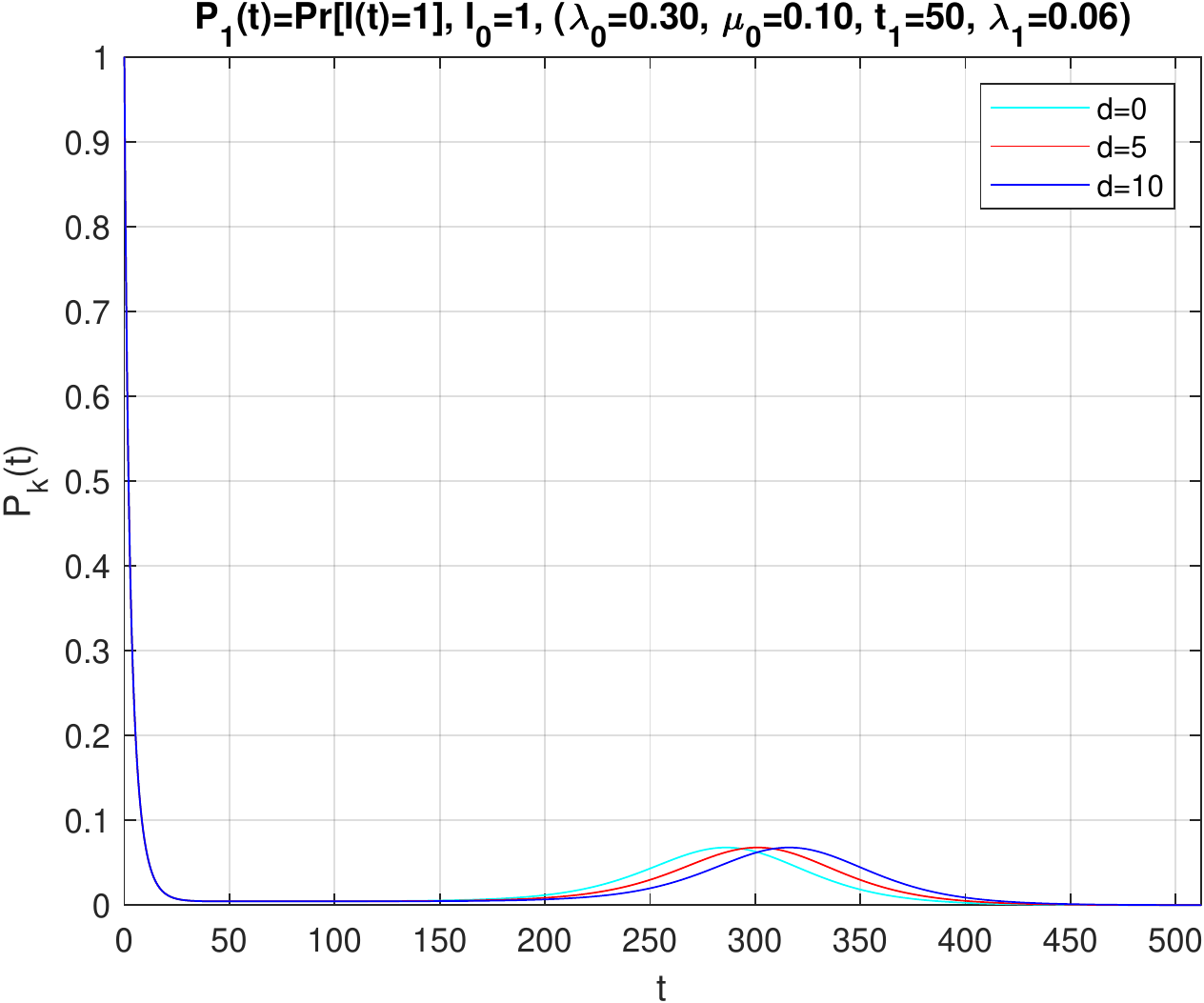}
\caption{\sf $P_1(t)$.}
\label{fig:Non-homo-BD-P_1(t)-I_0=1}
\end{minipage}
\hspace{0.5cm}
\begin{minipage}[t]{0.30\textwidth}
\centering
\includegraphics[width=\textwidth]{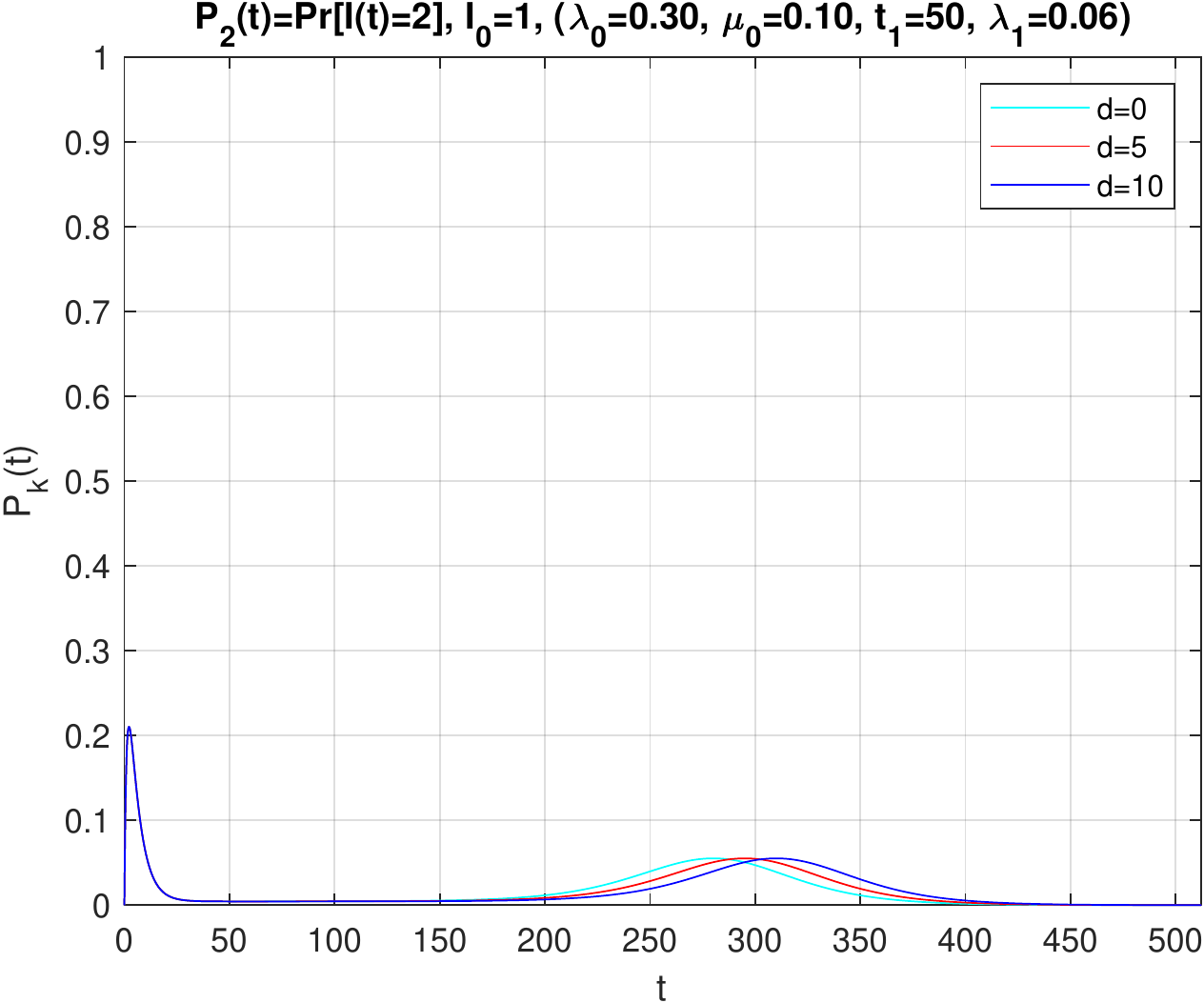}
\caption{\sf $P_2(t)$.}
\label{fig:Non-homo-BD-P_2(t)-I_0=1}
\end{minipage}
\end{figure}
\begin{figure}[hbt]
\begin{minipage}[h]{0.30\textwidth}
\centering
\includegraphics[width=\textwidth]{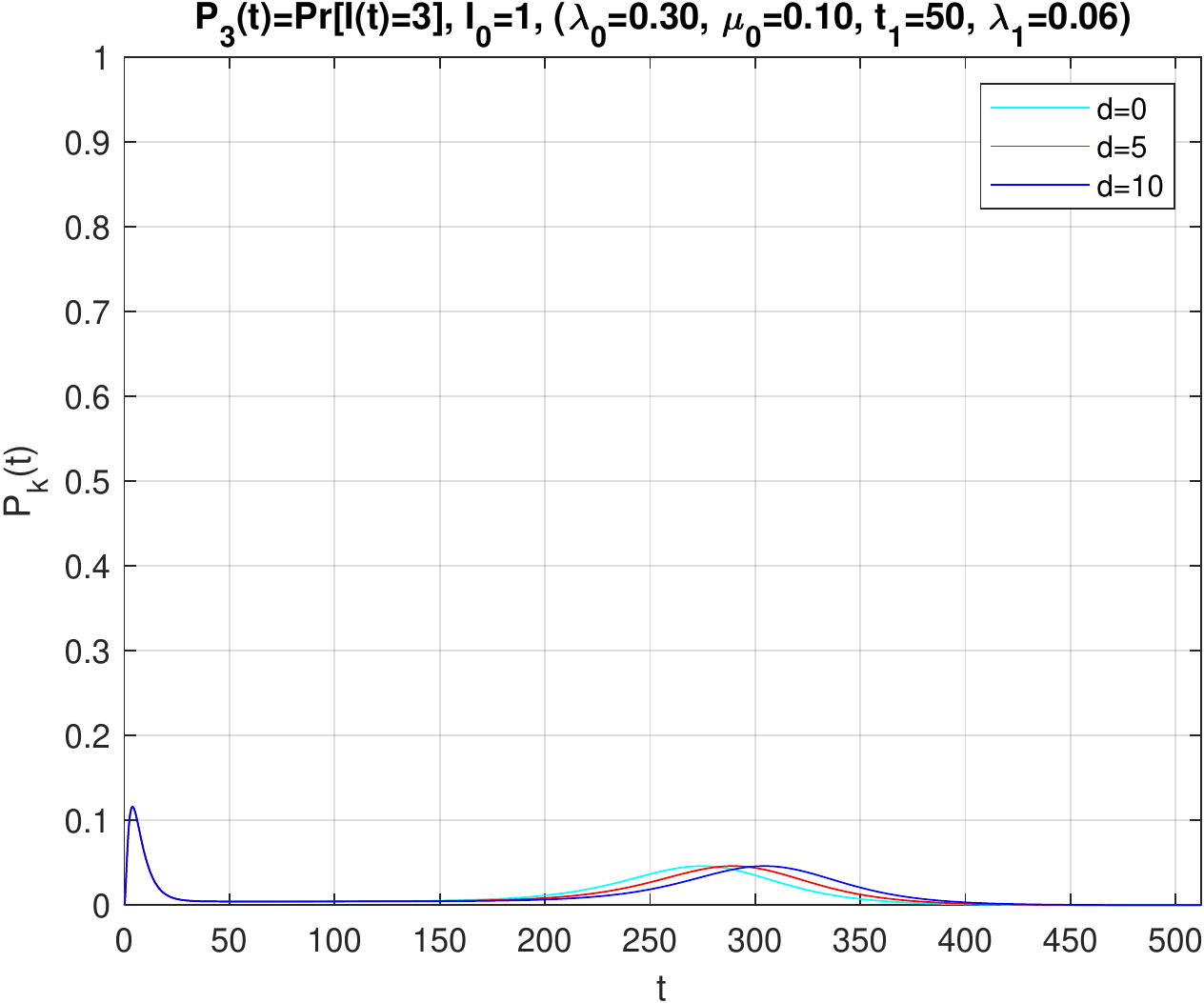}
\caption{\sf  $P_3(t)$.}
\label{fig:Non-homo-BD-P_3(t)-I_0=1}
\end{minipage}
\hspace{0.5cm}
\begin{minipage}[h]{0.30\textwidth}
\centering
\includegraphics[width=\textwidth]{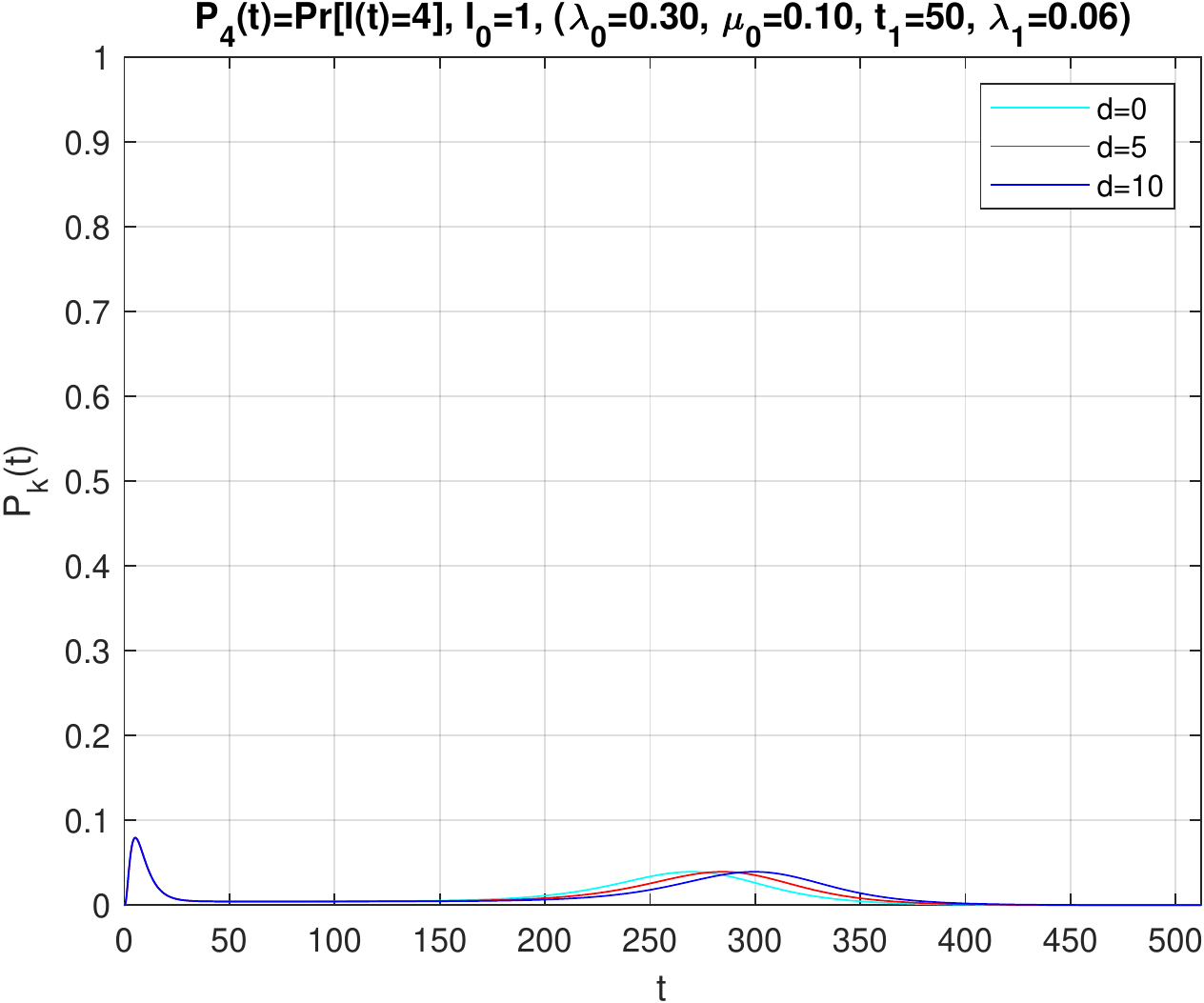}
\caption{\sf $P_4(t)$.}
\label{fig:Non-homo-BD-P_4(t)-I_0=1}
\end{minipage}
\hspace{0.5cm}
\begin{minipage}[h]{0.30\textwidth}
\centering
\includegraphics[width=\textwidth]{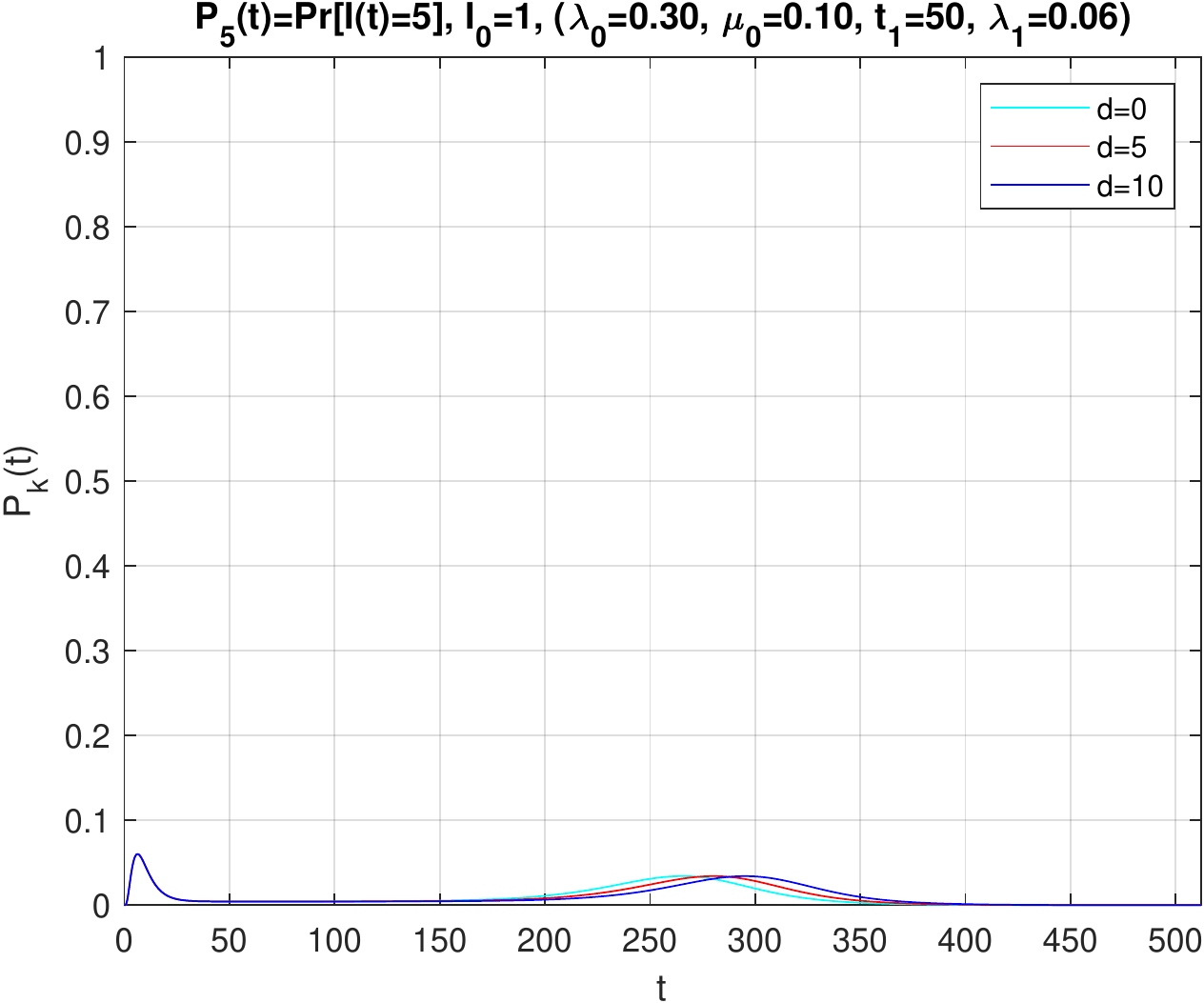}
\caption{\sf $P_5(t)$.}
\label{fig:Non-homo-BD-P_5(t)-I_0=1}
\end{minipage}
\end{figure}
\begin{figure}[bth]
\begin{minipage}[b]{0.30\textwidth}
\centering
\includegraphics[width=\textwidth]{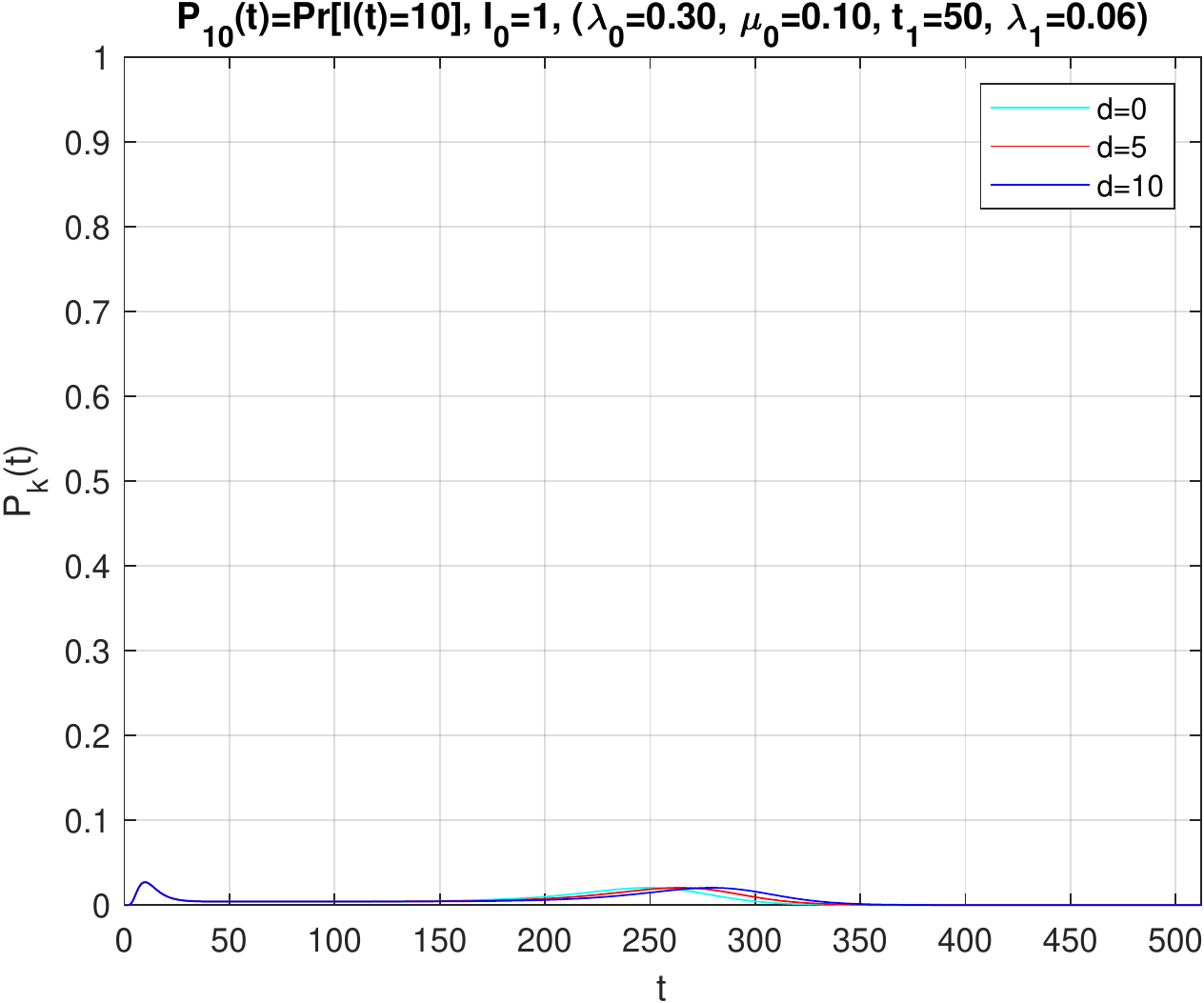}
\caption{\sf $P_{10}(t)$.}
\label{fig:Non-homo-BD-P_10(t)-I_0=1}
\end{minipage}
\end{figure}
\clearpage

\item \textbf{Mesh Plot of $Z(k,t)=P_k(t)$}:
Shown in Figure \ref{fig:Non-homo-BD-Z(k,t)-I_0=1} is a bird eye view of the time-dependent function $P_k(t)$ over the $(k, t)$ coordinate. This plot summarizes the various plots presented above.

\begin{figure}
\begin{minipage}[b]{0.80\textwidth}
\centering
\includegraphics[width=\textwidth]{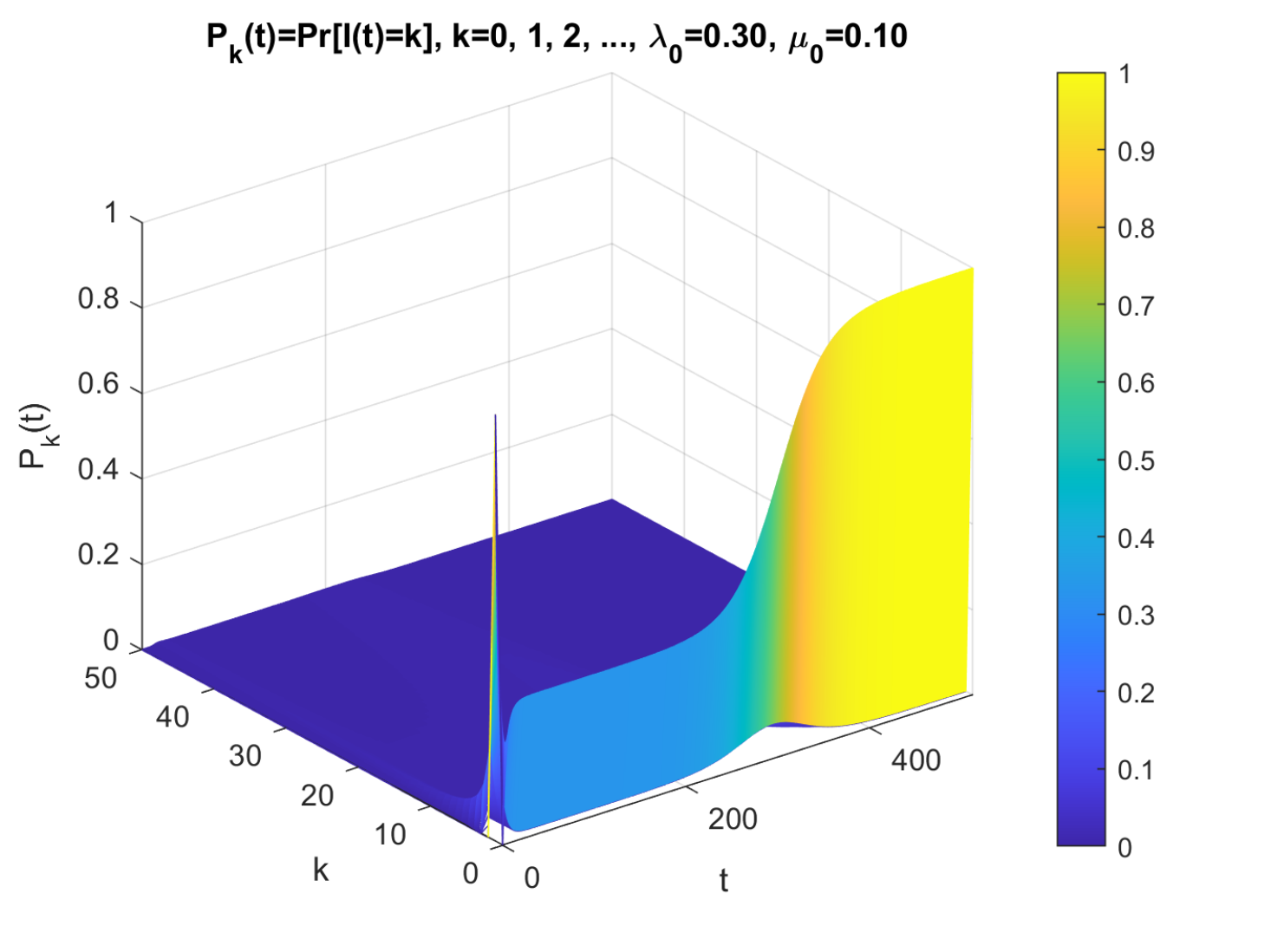}
\caption{\sf  $Z(k,t)=P_k(t)$ over the $(k t)$ plane}
\label{fig:Non-homo-BD-Z(k,t)-I_0=1}
\end{minipage}
\end{figure}

\end{enumerate}

\section{Discussion and Future Plans}

In the present article we presented a theoretical analysis of our stochastic model of an infectious disease based on the time-nonhomogeneous BD and BDI processes.

\begin{enumerate}

\item  First, we discussed a time-nonhomegeneous deterministic model, from which the expectation of stochastic processes of our interest (i.e., $\oI(t), \oA(t), \oB(t), \oR(t)$) were obtained. Among the three model parameter functions $\lambda(t), \mu(t)$ and $\nu(t)$, the difference of the first two, i.e., $a(t)=\lambda(t)-\mu(t)$, is of primary importance, and its integrated function $s(t)=\int_0^t a(u)\,du$, and its exponentiation play central roles in the analysis. 

\item We presented a hypothetical scenario, in which a government declares a state of emergency, requesting its public to significantly reduce their activities that may incur infections.  We showed how even a small delay in implementing the new order will result in a further increase in infections for a while, and prolong the period until $\oI(t)$ diminishes to practically zero. 

\item The analysis of the stochastic version showed that the BD process without immigration (i.e., no external arrival of the infected from the outside) is solvable exactly. An exact analysis of time-nonhomogeneour BDI process model remains an open question. 

\item The function $\Sigma(t)=\int_0^t e^{-s(u)}\,du$, is the most important determines the timing and duration of the transition in the function $\alpha(t)$ defined by (\ref{functions-alpha-beta}). The function $\alpha(t)$ is  directly related to $P_0(t)$, the probability that the infection comes to a halt by time $t$. 

\item From the graphical plots of the time-dependent function $P_k(t)$ presented in the last section reinforce our earlier claim (see \cite{kobayashi:2020b} Abstract) that it would be a futile effort to attempt to identify all possible reasons why environments of similar situations differ so much from in terms of epidemic patterns and the number of casualties. Mere ``luck" or ``chances" play a more significant role than most off us are led to believe.
Thus, we should be prepared for a worst possible scenarios, which only a stochastic model can provide with probabilistic qualification, such as e.g., a 95\% confidence level.

\end{enumerate}

Our future research plan include

\begin{itemize}
\item[(i)] Conduct simulation experiments for time-nonhomogeneous BD and BDI process models to empirically validate the analytic results presented in this article.  The simulation will be done in a fashion similar to what was reported in Part II \cite{kobayashi:2020b} on time-homogeneous cases.

\item[(ii)] Develop a method to estimate the model parameters from observable data in simulations.  After validating its effectiveness, we should apply the method to real data.  We expect that the estimated model parameters could be used to predict the near-term behavior of the infection process $I(t)$.  

\item[(iii)] Thus far, we have primarily dealt with the process $I(t)$, which is a Markov process.  But the processes $B(t)$ and $R(t)$ are not Markov processes. This makes it difficult to obtain their PDFs in an exact form.  An approximate analysis based on saddle-point integration will be investigated.

\item[(iv)] The main advantage of our stochastic modelling approach is that it not only allows us to better understand the stochastic behavior of an infectious disease, but also permits us to generalize the results obtained thus far to more realistic situations, because our model is intrinsically linear.  One important extension will be to deal with different types of infectious diseases, and different types of infectious and infected individuals. These situations can be adequately represented by introducing different ``classes" of susceptible population, and multiple "types" of infectious diseases.  

\end{itemize}

\section*{Acknowledgments}
\addcontentsline{toc}{section}{Acknowledgments}
The author thanks Prof. Brian L. Mark of George Mason University and Dr. Pei Chen of Qualcomm for their advice and help in use of MATLAB.
\bibliographystyle{ieeetr}
\bibliography{infections}

\end{document}